\documentclass[11pt]{article}%
\usepackage{amsmath}
\usepackage{amsfonts}
\usepackage{amssymb}
\usepackage{graphicx}
\usepackage{fullpage}
\providecommand{\U}[1]{\protect\rule{.1in}{.1in}}
\newtheorem{theorem}{Theorem}
\newtheorem{lemma}{Lemma}

\newenvironment{proof}[1][Proof]{\noindent\textbf{#1.} }{\ \rule{0.5em}{0.5em}}
\begin{document}

\title{\textbf{Strong converse rates for classical communication over thermal and
additive noise bosonic channels}}
\author{Bhaskar Roy Bardhan\\\textit{{\small {Hearne Institute for Theoretical Physics,}}} \\\textit{{\small {Department of Physics and Astronomy,}}} \\\textit{{\small {Louisiana State University,}}} \\\textit{{\small {Baton Rouge, Louisiana 70803, USA}}}
\and Mark M. Wilde\\\textit{{\small {Hearne Institute for Theoretical Physics,}}} \\\textit{{\small {Department of Physics and Astronomy,}}} \\\textit{{\small {Center for Computation and Technology,}}} \\\textit{{\small {Louisiana State University,}}} \\\textit{{\small {Baton Rouge, Louisiana 70803, USA}}}}
\date{\today}
\maketitle

\begin{abstract}
We prove that several known upper bounds on the classical capacity of thermal
and additive noise bosonic channels are actually strong converse rates. Our
results strengthen the interpretation of these upper bounds, in the sense that
we now know that the probability of correctly decoding a classical message
rapidly converges to zero in the limit of many channel uses if the
communication rate exceeds these upper bounds. In order for these theorems to
hold, we need to impose a maximum photon number constraint on the states input
to the channel (the strong converse property need not hold if there is only a
mean photon number constraint). Our first theorem demonstrates that Koenig and
Smith's upper bound on the classical capacity of the thermal bosonic channel
is a strong converse rate, and we prove this result by utilizing the
structural decomposition of a thermal channel into a pure-loss channel
followed by an amplifier channel. Our second theorem demonstrates that
Giovannetti \textit{et al}.'s upper bound on the classical capacity of a
thermal bosonic channel corresponds to a strong converse rate, and we prove
this result by relating success probability to rate, the effective dimension
of the output space, and the purity of the channel as measured by the R\'enyi
collision entropy. Finally, we use similar techniques to prove that similar
previously known upper bounds on the classical capacity of an additive noise
bosonic channel correspond to strong converse rates.

\end{abstract}

\section{Introduction}

A principal goal of quantum information theory is to understand the
transmission of classical data over many independent uses of a noisy quantum
channel. We say that a fixed rate of communication is achievable if for every
$\varepsilon>0$ there exists a coding scheme using the channel a sufficiently
large number of times such that its error probability is no larger than
$\varepsilon$. The maximum achievable rate for a given channel is known as the
classical capacity of the channel~\cite{Hol98,SW97}.

According to the above definition of capacity, there cannot exist an
error-free communication scheme if its rate exceeds capacity. Such a statement
is known as a ``weak converse theorem,'' and even though it establishes
capacity as a threshold, it suggests that it might be possible for one to
increase the communication rate $R$ by allowing for some error $\varepsilon>0$
whenever $R$ exceeds the capacity. However, a strong converse theorem (if it
holds) demonstrates that there is no such room for a trade-off between rate
and error in the limit of many independent uses of the channel (see
Figure~\ref{Fig1}(a) for a conceptual illustration of this idea). That is, a
strong converse theorem establishes capacity as a very sharp threshold, so
that it is guaranteed that the error probability of any communication scheme
converges to one in the limit of many independent channel uses if its rate
exceeds capacity. A strong converse theorem holds for the classical capacity
of all classical channels \cite{Wolfowitz1964,Arimoto}, and a number of works
have now established strong converse theorems for the classical capacity of
certain quantum channels~\cite{W99,Ogawa,KW09,Entanglementbreaking}. Recently,
a strong converse theorem has been proved to hold for the classical capacity
of the pure-loss bosonic channel~\cite{StrongConversePureLoss}.

The present paper considers the transmission of classical data over two
bosonic channels: the thermal noise channel and the additive noise channel. In
particular, we are interested in determining sharp thresholds for
communication over them, in the strong converse sense mentioned above. Both of
these channels are important models for understanding the ultimate
information-carrying capacity of electromagnetic waves and have been
investigated extensively
\cite{HW01,GGLMSY04,GGLMS04,GLMS04,HG12,KS2,KS1,KS3,PracticalPurposes,Lupo2011}. In the
thermal noise channel, the environment begins in a thermal equilibrium and the
channel mixes these noise photons with the signaling photons. More
specifically, this channel is modeled as a beamsplitter with transmissivity
$\eta$ which mixes the signaling photons (with average photon number $N_{S}$)
with a thermal state of average photon number $N_{B}$. In the additive noise
channel, each signal mode is randomly displaced in phase space according to a
Gaussian distribution \cite{HO93,Hall94}. Interestingly, the additive noise
bosonic channel can be obtained as a limiting case of the thermal noise
channel in which $\eta\rightarrow1$ and $N_{B} \rightarrow\infty$, with
$(1-\eta)N_{B} \rightarrow\bar{n}$, where $\bar{n}$ is the variance of the
noise introduced by the additive noise channel~\cite{GGLMS04}. This relation
allows for extending many results regarding the thermal channel to the
additive noise channel.

In this paper, we prove that several previously known upper bounds on the
classical capacity of these channels are actually ``strong converse rates''
\cite{GGLMS04,KS3,PracticalPurposes}. This means that the probability of
successfully decoding a classical message converges exponentially fast to zero
in the limit of many channel uses if the rate $R$ of communication exceeds
these strong converse rates. Previous work
\cite{GGLMS04,KS3,PracticalPurposes} has established that these upper bounds
are ``weak converse rates,'' meaning that there cannot be any error-free
communication scheme if the rate $R$ of communication exceeds them. Having an
upper bound serve as only a weak converse rate $R_{W}$ suggests that it might
be possible for one to increase the communication rate $R$ by allowing for
some error $\varepsilon>0$ whenever $R > R_{W}$. Our work here demonstrates
that there is no such room for a trade-off between rate and error in the limit
of many independent uses of the channel (see Figure~\ref{Fig1}(b) for a
conceptual illustration of this idea). Thus, our work strengthens the
interpretation of the upper bounds from
\cite{GGLMS04,GLMS04,KS3,PracticalPurposes}.

\begin{figure}[ptb]
\centering
\includegraphics[width=\columnwidth]{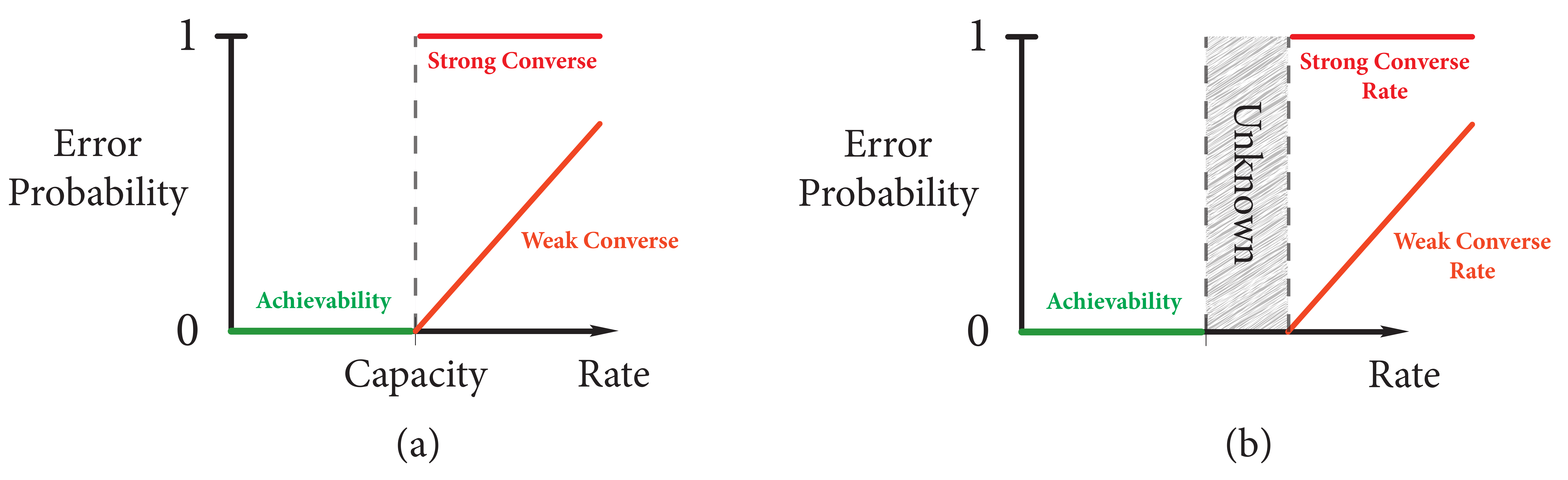}\caption{(Color online) Weak converse versus
strong converse rates for channels in which (a) the classical capacity is
exactly known (such as the pure-loss bosonic
channel~\cite{GGLMSY04,StrongConversePureLoss}), and (b) only lower bounds and
upper bounds are known, while the exact classical capacity is not known (such
as the thermal channel). In both cases, the figures illustrate the idea that
the error probability converges to one in the limit of many channel uses if a
communication rate corresponds to a strong converse rate, whereas establishing
a communication rate as a weak converse rate suggests that there exists room for a
trade-off between communication rate and error proability. Achievable rates
are such that there exists a communication scheme whose error probability
converges to zero in the limit of many channel uses.}%
\label{Fig1}%
\end{figure}

\section{Summary of results}

We now give a brief summary of the present paper's two main contributions:

\begin{enumerate}
\item {Following \cite{StrongConversePureLoss}, we begin by showing that a
strong converse theorem need not hold for the classical capacity of the
thermal noise channel and the additive noise channel whenever there is only a
\textit{mean photon number constraint}. }

\item {In light of the above observation and again following
\cite{StrongConversePureLoss}, we impose instead a maximum photon number
constraint, in such a way that nearly all of the ``shadow'' of the average
density operator for a given code is required to be on a subspace with photon
number no larger than a particular number, so that the shadow outside this
subspace vanishes in the limit of many channel uses. Under such a maximum
photon number constraint, we demonstrate that several previously known upper
bounds \cite{GGLMS04,GLMS04,KS3,PracticalPurposes} on the classical capacity
of the thermal and additive noise bosonic channels correspond to strong
converse rates.}
\end{enumerate}

The present paper is organized as follows. In Section~\ref{Prelims}, we review
several preliminary ideas, including mathematical definitions of the thermal
noise channel and additive noise channels and structural decompositions of
them that are useful in our work. We also present the basic notions of the
quantum R\'enyi entropy and its relation with smooth min-entropy
\cite{RWISIT1}. Section~\ref{Meanphotonnumber} illustrates a simple proof that
the strong converse need not hold with only a mean photon number constraint
for the above two bosonic channels, following the approach given in
\cite{StrongConversePureLoss}. In Section~\ref{SCT}, we instead impose a
maximum photon number constraint and prove that several previously known upper
bounds \cite{GGLMS04,GLMS04,KS3,PracticalPurposes} on the classical capacity
of the thermal noise channel and the additive noise channel are actually
strong converse rates.
Section~\ref{sec:conclusion} contains some concluding remarks and an outlook
for future research, in particular the implications of our results for other
noisy bosonic channels.

\section{Preliminaries}

\label{Prelims}

\subsection{Thermal noise channel}

\label{sec:thermal}

The thermal noise channel is represented by a Gaussian completely positive
trace-preserving (CPTP) map---i.e., it evolves Gaussian input states to
Gaussian output states \cite{WPGCRSL12}. The thermal channel can be modeled by
a beamsplitter of transmissivity $\eta$ that couples the input signal with a
thermal state of mean photon number $N_{B}$. The parameter $\eta\in[0,1]$
characterizes the fraction of input photons that make it to the output on
average. The special case $N_{B}=0$ (zero-temperature reservoir) corresponds
to the pure-loss bosonic channel $\mathcal{E}_{\eta,0}$, in which each input
photon has a probability $\eta$ of reaching the output.

The beamsplitter transformation corresponding to the thermal channel can be
written as the following Heisenberg evolution of the signal mode
operator~$\hat{a}$ and the environmental mode operator~$\hat{b}$:
\begin{align}
\hat{a}  &  \longrightarrow\sqrt{\eta}\hat{a}+\sqrt{1-\eta}\hat{b} \,
,\nonumber\\
\hat{b}  &  \longrightarrow-\sqrt{1-\eta}\hat{a} + \sqrt{\eta}\hat{b}\, .
\label{eq:bs-trans}%
\end{align}
Tracing out the environmental mode $\hat{b}$ yields the following CPTP map
$\mathcal{E}_{\eta,N_{B}}$ for the thermal noise channel:
\begin{equation}
\mathcal{E}_{\eta,N_{B}}=\text{Tr}_{\hat{b}}\left[  U (\rho_{a} \otimes
\rho_{b}) U^{\dagger}\right]  ,
\end{equation}
where $\rho_{a}$ and $\rho_{b}$ correspond to the input state and the
environmental thermal state, respectively, and the unitary $U$ can be inferred
from the transformation in (\ref{eq:bs-trans}). The thermal state $\rho_{b}$
is equivalent to an isotropic Gaussian mixture of coherent states with average
photon number $N_{B}\geq 0$ \cite{GK04}:
\begin{equation}
\rho_{b}=\int d^{2} \alpha\, \frac{\exp(-|\alpha|^{2}/N_{B})}{\pi N_{B}} \,
|\alpha\rangle\langle\alpha|=\frac{1}{(N_{B}+1)} \sum_{l=0}^{\infty} \left(
\frac{N_{B}}{N_{B}+1}\right)  ^{l} | l \rangle\langle l | .
\end{equation}
As an example, we can see that a vacuum state at the input of the thermal
channel produces the following thermal state at the output:
\begin{equation}
\label{outputthermalvac}\mathcal{E}_{\eta,N_{B}} (|0 \rangle\langle
0|)=\frac{1}{\left(  (1-\eta)N_{B}+1\right)  } \sum_{l=0}^{\infty} \left(
\frac{(1-\eta)N_{B}}{(1-\eta)N_{B}+1}\right)  ^{l}| l \rangle\langle l | .
\end{equation}

Despite extensive efforts to find the classical capacity of the thermal
channel \cite{HW01,GGLMSY04,GGLMS04,GLMS04,KS2,KS1,KS3,PracticalPurposes}, it
is still unknown. However, a few upper and lower bounds on it are now known.
Holevo and Werner have shown that the classical capacity $C(\mathcal{E}%
_{\eta,N_{B}})$ of the thermal noise channel $\mathcal{E}_{\eta,N_{B}}$
satisfies~\cite{HW01}
\begin{equation}
C(\mathcal{E}_{\eta,N_{B}}) \geq g(\eta N_{S}+(1-\eta) N_{B})-g((1-\eta
)N_{B}), \label{eq:LB}%
\end{equation}
where
\begin{equation}
g(x) \equiv(x+1) \log_{2} (x+1)-x \log_{2} x
\end{equation}
is the entropy of a bosonic thermal state with mean photon number $x$. They
established this lower bound by proving that coherent-state coding schemes
achieve the communication rate on the RHS of (\ref{eq:LB}). It has been
conjectured that the above lower bound is equal to the classical capacity of
the thermal noise channel, provided that a certain minimum output entropy
conjecture is true~\cite{GGLMS04}. The results of
\cite{GGLMS04,GLMS04} establish the following upper bound on the classical
capacity of the thermal bosonic channel:\footnote{The fact that the results of
\cite{GGLMS04,GLMS04} give upper bounds on the classical capacity of the
thermal channel was recently communicated in \cite{PracticalPurposes}.}
\begin{equation}
\label{upperboundGV}C(\mathcal{E}_{\eta,N_{B}}) \le g(\eta N_{S}+(1-\eta)
N_{B})-\log_{2}(1+2(1-\eta)N_{B}).
\end{equation}
This upper bound lies within 1.45 bits of the lower bound in (\ref{eq:LB}).
Koenig and Smith determined tight upper bounds on the classical capacity of
the thermal noise channel whenever $\eta=1/2$~\cite{KS1,KS2}, by proving a
quantum entropy power inequality. They also established the following upper
bound on the classical capacity $C(\mathcal{E}_{\eta,N_{B}})$~\cite{KS3}:
\begin{equation}
C(\mathcal{E}_{\eta,N_{B}}) \leq g(\eta N_{S}/[(1-\eta)N_{B}+1]).
\label{eq:KS-upper-bound}%
\end{equation}
This latter bound is also within 1.45 bits of the lower bound in (\ref{eq:LB}).

In this paper, we show that both of the upper bounds in (\ref{upperboundGV})
and (\ref{eq:KS-upper-bound}) correspond to strong converse rates.

\subsection{Additive noise channel}

\label{sec:additive}

The additive noise channel is specified by the following CPTP map:
\begin{equation}
\label{classical}\mathcal{N}_{\bar{n}} (\rho) \equiv\int d^{2} \alpha\,
P_{\bar{n}} (\alpha) \, D(\alpha)\rho D^{\dagger}(\alpha),
\end{equation}
where $P_{\bar{n}} (\alpha)= \exp\left(  -|\alpha|^{2}/\bar{n}\right)  /
(\pi\bar{n})$ and $D(\alpha)\equiv\exp(\alpha\hat{a}^{\dagger}-\alpha^{*}
\hat{a})$ is a displacement operator for the input signal mode $\hat{a}$. The
Gaussian probability distribution $P_{\bar{n}} (\alpha)$ determines the random
displacement of the signal mode $\hat{a}$ in phase space. The variance
$\bar{n}$ of this distribution completely characterizes the additive noise
channel $\mathcal{N}_{\bar{n}} $, and it represents the number of noise
photons added to the mode $\hat{a}$ by the channel~\cite{Hall94}. For $\bar
{n}=0$, the CPTP map in (\ref{classical}) reduces to the identity channel,
while for $\bar{n}>0$, noise photons are injected into the channel. As an
example, we can see that the action of the classical noise channel
$\mathcal{N}_{\bar{n}}$ on a vacuum-state input produces a thermal-state
output:
\begin{equation}
\label{outputadditivevac}\mathcal{N}_{\bar{n}} (|0 \rangle\langle0|)=\frac
{1}{\bar{n}+1} \sum_{l=0}^{\infty}\left(  \frac{\bar{n}}{\bar{n}+1}\right)
^{l} \vert l \rangle\langle l \vert.
\end{equation}

Since the additive noise channel can be obtained from the thermal noise
channel in the limit $\eta\rightarrow1$ and $N_{B} \rightarrow\infty$, with
$(1-\eta)N_{B} \rightarrow\bar{n}$~\cite{GGLMS04} (see also
Appendix~\ref{symplectic} for a review of this), many results regarding the
thermal channel apply to the additive noise channel as well. For example, we
obtain the following bounds on the classical capacity of the additive noise
channel $\mathcal{N}_{\bar{n}}$ simply by taking the aforementioned limit in
(\ref{eq:LB}), (\ref{upperboundGV}), and (\ref{eq:KS-upper-bound}),
respectively:
\begin{align}
C(\mathcal{N}_{\bar{n}})  &  \geq g(N_{S}+\bar{n})-g(\bar{n}) \,
,\label{eq:additiveLB}\\
C(\mathcal{N}_{\bar{n}})  &  \leq g(N_{S}+\bar{n})-\log_{2}(1+2 \bar{n}) \,
,\label{additiveUB2}\\
C(\mathcal{N}_{\bar{n}})  &  \leq g(N_{S}/[\bar{n}+1]) \, .
\label{additiveUB1}%
\end{align}
This last bound easily follows from (\ref{eq:KS-upper-bound}), but as far as
we can tell, it appears to be new.

In this paper, we prove that both of the upper bounds in (\ref{additiveUB2})
and (\ref{additiveUB1}) correspond to strong converse rates.

\subsection{Structural decompositions}

\label{sec:decompositions}

Both the thermal and additive noise channels can be decomposed as a
concatenation of other channels \cite{SGGLMY04,GGLMS04,CGH06,GNLSC12}, and
these decompositions are helpful in establishing upper bounds on capacity. We
briefly review these decompositions in this section and, for convenience, give
a full proof of them in Appendix~\ref{symplectic} using the symplectic
formalism \cite{HW01,WPGCRSL12}. The thermal noise channel $\mathcal{E}%
_{\eta,N_{B}}$ can be regarded as the application of the additive noise
channel $\mathcal{N}_{(1-\eta) N_{B}}$ to the output of the pure-loss bosonic
channel $\mathcal{E}_{\eta,0}$~\cite{GGLMS04}:
\begin{equation}
\label{compositionI}\mathcal{E}_{\eta,N_{B}} (\rho)=\left(  \mathcal{N}%
_{(1-\eta) N_{B}} \circ\mathcal{E}_{\eta,0}\right)  (\rho) .
\end{equation}
The following alternative composition rule holds for the thermal noise
channel~\cite{CGH06} (see also \cite{GNLSC12} and \cite{KS3}):
\begin{equation}
\mathcal{E}_{\eta,N_{B}}(\rho)=\left(  \mathcal{A}_{G_{1}} \circ
\mathcal{E}_{\eta_{1},0}\right)  (\rho) , \label{eq:compositionII}%
\end{equation}
where $\mathcal{A}_{G_{1}}$ is an amplifier channel with gain $G_{1}%
=(1-\eta)N_{B}+1$ and $\mathcal{E}_{\eta_{1},0}$ is the pure-loss bosonic
channel with transmissivity $\eta_{1}=\eta/G_{1}$. This means that the thermal
noise channel $\mathcal{E}_{\eta,N_{B}}$ can be viewed as a cascade of the
above two channels, in which the input state is propagated through the
pure-loss bosonic channel and followed by the amplifier channel. Taking the
limits $\eta\rightarrow1$ and $N_{B} \rightarrow\infty$, with $(1-\eta)N_{B}
\rightarrow\bar{n}$, we obtain from (\ref{eq:compositionII}) the following
composition rule for the additive noise channel \cite{SGGLMY04}:
\begin{equation}
\label{c3}\mathcal{N}_{\bar{n}}(\rho)=(\mathcal{A}_{\bar{n}+1} \circ
\mathcal{E}_{\frac{1}{\bar{n}+1},0})(\rho)\, .
\end{equation}

\subsection{Quantum R\'enyi entropy and smooth min-entropy}

The quantum R\'{e}nyi entropy $H_{\alpha}(\rho)$ of a density operator $\rho$
is defined for $0<\alpha<\infty$, $\alpha\neq1$ as
\begin{equation}
H_{\alpha}(\rho)\equiv\frac{1}{1-\alpha}\log_{2}\operatorname{Tr}[\rho
^{\alpha}]\,.
\end{equation}
It is a monotonic function of the \textquotedblleft$\alpha$%
-purity\textquotedblright\ $\operatorname{Tr}[\rho^{\alpha}]$, and the von
Neumann entropy $H(\rho)$ is recovered from it in the limit $\alpha
\rightarrow1$:%
\[
\lim_{\alpha\rightarrow1}H_{\alpha}(\rho)=H(\rho)\equiv-\operatorname{Tr}%
[\rho\log_{2}\rho]\,.
\]
The min-entropy is recovered from it in the limit as $\alpha\rightarrow\infty
$:%
\[
\lim_{\alpha\rightarrow\infty}H_{\alpha}(\rho)=H_{\min}\left(  \rho\right)
\equiv-\log_{2}\left\Vert \rho\right\Vert _{\infty},
\]
where $\left\Vert \rho\right\Vert _{\infty}$ is the infinity norm of $\rho$.
For an additive noise channel $\mathcal{N}_{\bar{n}}$, the R\'{e}nyi entropy
$H_{\alpha}(\mathcal{N}_{\bar{n}}(\rho))$ for $\alpha\in\{2,3,\ldots\}$
achieves its minimum value when the input $\rho$ is the vacuum state
$|0\rangle$~\cite{GLMS04}:
\begin{equation}
\min_{\rho}H_{\alpha}(\mathcal{N}_{\bar{n}}(\rho))=H_{\alpha}(\mathcal{N}%
_{\bar{n}}(|0\rangle\langle0|))=\frac{\log_{2}[(\bar{n}+1)^{\alpha}-\bar
{n}^{\alpha}]}{\alpha-1}\,.
\end{equation}
Similarly, for the thermal noise channel $\mathcal{E}_{\eta,N_{B}}$, the
R\'{e}nyi entropy $H_{\alpha}(\mathcal{E}_{\eta,N_{B}}(\rho))$ for $\alpha
\in\{2,3,\ldots\}$ achieves its minimum value when the input $\rho$ is the
vacuum state $|0\rangle$~\cite{GLMS04}:
\begin{equation}
\min_{\rho}H_{\alpha}(\mathcal{E}_{\eta,N_{B}}(\rho))=H_{\alpha}%
(\mathcal{E}_{\eta,N_{B}}(|0\rangle\langle0|))=\frac{\log_{2}[((1-\eta
)N_{B}+1)^{\alpha}-((1-\eta)N_{B})^{\alpha}]}{\alpha-1}\,.
\end{equation}
One of the most important questions in quantum information theory is whether
the vacuum input still gives the minimum output R\'enyi entropy for other
values of $\alpha$ (with the case $\alpha=1$ being of especial importance
\cite{GGLMS04}).

An elegant generalization of the R\'{e}nyi entropy is the smooth R\'{e}nyi
entropy, first introduced by Renner and Wolf for classical probability
distributions~\cite{RS}. The results there were further generalized to the
quantum case (density operators) by considering the set $\mathcal{B}%
^{\varepsilon}(\rho)$ of density operators $\tilde{\rho}$ that are
$\varepsilon$-close to $\rho$ in trace distance for $\varepsilon\geq
0$~\cite{RennerThesis}. The $\varepsilon$-smooth quantum R\'{e}nyi entropy of
order $\alpha$ of a density operator $\rho$ is defined as~\cite{RennerThesis}%
\begin{equation}
H_{\alpha}^{\varepsilon}(\rho)\equiv\left\{
\begin{array}
[c]{cc}%
\inf_{\tilde{\rho}\in\mathcal{B}^{\varepsilon}(\rho)}H_{\alpha}(\tilde{\rho
}) & 0\leq\alpha<1\\
\sup_{\tilde{\rho}\in\mathcal{B}^{\varepsilon}(\rho)}H_{\alpha}(\tilde{\rho
}) & 1<\alpha<\infty
\end{array}
\right.  .
\end{equation}
In the limit as $\alpha\rightarrow\infty$, we recover the smooth min-entropy
of $\rho$ \cite{RennerThesis,Tomamichelthesis}:%
\begin{equation}
H_{\min}^{\varepsilon}(\rho)\equiv\sup_{\tilde{\rho}\in\mathcal{B}%
^{\varepsilon}(\rho)}\left[  -\log_{2}\left\Vert \widetilde{\rho}\right\Vert
_{\infty}\right]  \,. \label{normsmoothing}%
\end{equation}
From the above, we see that the following relation holds%
\begin{equation}
\inf_{\widetilde{\rho}\in\mathcal{B}^{\varepsilon}\left(  \rho\right)
}\left\Vert \widetilde{\rho}\right\Vert _{\infty}=2^{-H_{\min}^{\varepsilon
}\left(  \rho\right)  }\,.
\end{equation}
The following inequality is one of the main results of \cite{RS}, and it
demonstrates a connection between the smooth min-entropy and the R\'{e}nyi
entropy of order $\alpha>1$:%
\begin{equation}
H_{\min}^{\varepsilon}\left(  \rho\right)  \geq H_{\alpha}\left(  \rho\right)
-\frac{1}{\alpha-1}\log_{2}\left(  \frac{1}{\varepsilon}\right)  .
\label{Renyismoothing}%
\end{equation}
For convenience, Appendix~\ref{smoothentropy} reviews the proof of the above
inequality from \cite{RS}.

\subsection{Strong converse for the noiseless qubit channel}

For a noiseless qubit channel, the argument for the strong converse theorem is
rather simple~\cite{N99,KW09}, but it plays an important role in this work, so
we review it briefly. Suppose that any scheme for classical communication over
$n$ noiseless qubit channels consists of an encoding of the message $m$ as a
quantum state on $n$ qubits, followed by a decoding POVM $\{\Lambda_{m}\}$. The rate of the code is $R=(\log_{2}M)/n$, and the success
probability for a receiver to correctly recover the message is upper
bounded as%
\begin{align}
\frac{1}{M}\sum_{m}\text{Tr}\{\Lambda_{m}\rho_{m}\}  &  \leq\frac{1}{M}%
\sum_{m}\text{Tr}\{\Lambda_{m}\}||\rho_{m}||_{\infty}\nonumber\\
&  \leq\frac{1}{M}\sum_{m}\text{Tr}\{\Lambda_{m}\}\nonumber\\
&  =M^{-1}2^{n}\nonumber\\
&  =2^{-n(R-1)}\nonumber
\end{align}
In the above, we have used that the infinity norm $||\rho_{m}||_{\infty}$ is never larger
than one, and $\sum_{m}\Lambda_{m}=I^{\otimes n}$ for any POVM $\{\Lambda
_{m}\}$. The above argument demonstrates that for a rate $R>1$, the success
probability of any communication scheme decreases exponentially fast to zero
with increasing $n$.

The above proof for the noiseless qubit channel highlights an interplay of the
success probability of decoding with rate, the dimension of
the channel output space, and the purity of the channel (quantified by the
infinity norm of the output states of the channel). Our argument for the
thermal and additive noise channel can be viewed as a generalization of the above proof.

\section{No strong converse under a mean photon number constraint}

\label{Meanphotonnumber}

If we allow the input signal states to have an arbitrarily large number of
photons, then the classical capacity of the thermal noise channel is infinite
\cite{HW01}. Thus, in order to have a sensible notion of classical capacity
for this channel, we must impose some kind of constraint on the photon number
of the signaling states. A common constraint employed in the literature
\cite{HW01,GGLMSY04} is that the mean number of photons in any codeword
transmitted through the channel can be at most $N_{S}\geq0$ for each use of
the channel (mean photon number constraint). In this section, we show that a
strong converse does not hold for the classical capacity of the thermal noise
and additive noise bosonic channels under such a \emph{mean photon number
constraint}. The arguments presented here for proving this are the same as in
\cite{StrongConversePureLoss}.

In order to show a violation of the strong converse with a mean photon number
constraint, we consider the encoding of a classical message $m$ into $n$-mode
coherent-state codewords, where each codeword is independently sampled from an
isotropic complex Gaussian distribution with variance $P>N_{S}$%
~\cite{GGLMSY04,StrongConversePureLoss}. Let
\begin{equation}
|\alpha^{n}(m)\rangle\equiv|\alpha_{1}(m)\rangle\otimes\cdots\otimes
|\alpha_{n}(m)\rangle\label{encoding}%
\end{equation}
denote each of the $n$-mode coherent-state codewords, and we demand that every
such codeword in the codebook $\{|\alpha^{n}(m)\rangle\}_{m\in\lbrack M]}$ has
mean photon number $P>N_{S}$ (letting $M$ be the size of the message set and
denoting the message set by $[M]$).

The Holevo-Werner coding theorem \cite{HW01} states that if we choose these
codewords at a rate equal to $\frac{1}{n}\log_{2}\left(  M\right)  \approx g\left(
\eta P+\left(  1-\eta\right)  N_{B}\right)  -g\left(  \left(  1-\eta\right)
N_{B}\right)  $\ and transmit them over the thermal noise channel
$\mathcal{E}_{\eta,N_{B}}$, the receiver can decode them with arbitrarily
large success probability. That is, as long as $\frac{1}{n}\log_{2}\left(
M\right)  \approx g\left(  \eta P+\left(  1-\eta\right)  N_{B}\right)
-g\left(  \left(  1-\eta\right)  N_{B}\right)  $ and the number $n$ of channel
uses is sufficiently large, there exists a measurement $\left\{  \Lambda
_{m}\right\}  _{m\in\left[  M\right]  }$ and codebook such that%
\begin{equation}
\forall m\in\left[  M\right]  :\text{Tr}\left\{  \Lambda_{m}\mathcal{E}%
_{\eta,N_{B}}^{\otimes n}\left(  |\alpha^{n}(m)\rangle\langle\alpha
^{n}(m)|\right)  \right\}  \geq1-\varepsilon, \label{eq:bnd-succ-prob}%
\end{equation}
for $\varepsilon$ an arbitrarily small positive number.

\subsection{No strong converse with mixed-state codewords}

What we can do to violate the strong converse is to pick $P$ such that%
\[
g\left(  \eta P+\left(  1-\eta\right)  N_{B}\right)  -g\left(  \left(
1-\eta\right)  N_{B}\right)  >g(\eta N_{S}+(1-\eta)N_{B})-\log_{2}%
(1+2(1-\eta)N_{B}),
\]
where the term on the RHS\ is the upper bound from (\ref{upperboundGV}) on the
capacity of a thermal channel $\mathcal{E}_{\eta,N_{B}}$ in which we allow for
a mean photon number $N_{S}$. We then modify the codebook given above so that
each codeword has the following form:%
\begin{equation}
\rho_{m}\equiv(1-p)\left\vert \alpha^{n}(m)\right\rangle \left\langle
\alpha^{n}(m)\right\vert +p(|0\rangle\langle0|)^{\otimes n},
\label{eq:mixed-state-codewords}%
\end{equation}
with $(1-p)P=N_{S}$ (mean photon number constraint) and $0\leq p\leq1$.
Observe that the mean photon number of these modified codewords is equal to
$N_{S}$ so that we satisfy the mean photon number constraint. Transmitting
these codewords through the thermal noise channel gives the state%
\[
\mathcal{E}_{\eta,N_{B}}^{\otimes n}\left(  \rho_{m}\right)  .
\]
The success probability for correctly decoding the codewords with the decoding
measurement $\{\Lambda_{m}\}$ for the original code is then%
\begin{align}
\operatorname{Tr}\{\Lambda_{m}\mathcal{E}_{\eta,N_{B}}^{\otimes n}\left(
\rho_{m}\right)  \}  &  \geq(1-p)\operatorname{Tr}\{\Lambda_{m}\mathcal{E}%
_{\eta,N_{B}}^{\otimes n}\left(  \left\vert \alpha^{n}(m)\right\rangle
\left\langle \alpha^{n}(m)\right\vert \right)  \}\\
&  \geq(1-p)(1-\varepsilon). \label{successprobability}%
\end{align}
The last inequality follows from (\ref{eq:bnd-succ-prob}). The inequality in
(\ref{successprobability}) proves that the success probability need not
converge to zero in the limit of many channel uses for a rate larger than the
upper bound on the classical capacity from (\ref{upperboundGV}) under a mean
photon number constraint on each codeword in the codebook. A very similar
argument proves that the strong converse need not hold for the classical
capacity of the additive noise channel under only a mean photon number constraint.

\subsection{No strong converse with pure-state codewords}

In this section, we show that the the classical capacity for the thermal noise
and the additive noise channels need not obey a strong converse when only a
mean photon number constraint is imposed, even when restricting to pure-state
codewords. The argument is again similar to that in
{\cite{StrongConversePureLoss}. }This result is demonstrated by considering
the rate to be larger than the upper bound in (\ref{eq:KS-upper-bound}) for
the thermal noise channel.

We follow the arguments in {\cite{StrongConversePureLoss} and make use of an
ancillary single photon to \emph{purify} the mixed-state codewords in
(\ref{eq:mixed-state-codewords}) to be as follows:}%
\begin{equation}
|\gamma_{p}(m)\rangle\equiv\sqrt{(1-p)}\left\vert \alpha^{n}(m)\right\rangle
|0\rangle+\sqrt{p}|0\rangle^{\otimes n}|1\rangle.
\nonumber\label{codewordmixed}%
\end{equation}
{This additional mode has a negligible effect on the code parameters.}

Following {\cite{StrongConversePureLoss}, we can now show that the average
number of photons in each codeword above is equal to}%
\[
{\operatorname{Tr}\left\{  \left(  \frac{1}{n+1}\sum_{i=1}^{n+1}\hat{a}%
_{i}^{\dagger}\hat{a}_{i}\right)  |\gamma_{p}(m)\rangle\langle\gamma
_{p}(m)|\right\}  =(1-p)\frac{nP}{n+1}+p\frac{1}{n+1}.}%
\]
{Thus, we can set $p$ and $n$ such that the mean number of photons is equal to
$N_{S}$ (mean photon number constraint). It now follows by using the argument
in (\ref{successprobability}) that the success probability of correctly
decoding the message is larger than $(1-p)(1-\varepsilon)$ (the receiver
simply traces over the last mode and performs the POVM }$\left\{  \Lambda
_{m}\right\}  $){. This proves that a strong converse need not hold for the
classical capacity of the thermal channel under a mean photon number
constraint, even when restricting to pure-state codewords. Similar arguments
can be used to show a similar result for the additive noise channel.}

\section{Strong converse rates under a maximum photon number constraint}

\label{SCT}

In light of the results in the previous section, we can only hope that a
strong converse theorem holds under some alternative photon number constraint.
Let $\rho_{m}$ denote a codeword of any code that we wish to transmit through
the thermal noise channel $\mathcal{E}_{\eta,N_{B}}$ with transmissivity
$\eta$ and the noise photon number $N_{B}$. Following
\cite{StrongConversePureLoss}, we impose a \emph{maximum photon-number
constraint}, by demanding that the average code-density operator $\frac{1}%
{M}\sum_{m}\rho_{m}$ ($M$ is the total number of messages) has a large shadow
onto a subspace with photon number no larger than some fixed amount $nN_{S}$.
In more detail, we define a photon number cutoff projector $\Pi_{L}$
projecting onto a subspace of $n$ bosonic modes such that the total photon
number is no larger than $L$:
\begin{equation}
\Pi_{L}\equiv\sum_{a_{1},\ldots,a_{n}:\sum_{i}a_{i}\leq L}|a_{1}\rangle\langle
a_{1}|\otimes\ldots\otimes|a_{n}\rangle\langle a_{n}|,
\end{equation}
where $|a_{i}\rangle$ is a photon number state of photon number $a_{i}$. We
demand that the following maximum photon number constraint is satisfied%
\begin{equation}
\frac{1}{M}\sum_{m}\text{Tr}\left\{  \Pi_{\left\lceil nN_{S}\right\rceil }%
\rho_{m}\right\}  \geq1-\delta_{1}(n), \label{Constraint}%
\end{equation}
where $\delta_{1}(n)$ is a function that decreases to zero as $n$ increases.

A useful bound for us is that the rank of the photon number cutoff projector
$\Pi_{\lceil nN_{S}\rceil}$ cannot be any larger than $2^{n[g(N_{S})+\delta]}$
and $\delta\geq\frac{1}{n}(\log_2 e+\log_2(1+\frac{1}{N_{S}}))$ (Lemma 3 in
Ref.~\cite{StrongConversePureLoss}), i.e.,%

\begin{equation}
\text{Tr}\left\{  \Pi_{\lceil nN_{S}\rceil}\right\}  \leq2^{n[g(N_{S}%
)+\delta]}. \label{Rank}%
\end{equation}
The constant $\delta$ can be chosen to be an arbitrarily small positive
constant by taking $n$ large enough.

In what follows, we prove that several previously known upper bounds
\cite{GGLMS04,GLMS04,KS3,PracticalPurposes} on the classical capacity of the
thermal and additive noise channels are actually strong converse rates.

\subsection{Koenig-Smith bound is a strong converse rate for the thermal
channel}

\begin{theorem}
[(\ref{eq:KS-upper-bound}) is a strong converse rate for $\mathcal{E}%
_{\eta,N_{B}}$]\label{TheoremI} For the thermal noise channel $\mathcal{E}%
_{\eta,N_{B}}$, the average success probability of any code satisfying
(\ref{Constraint}) is bounded as follows:%
\[
\frac{1}{M}\sum_{m}\operatorname{Tr}\{\Lambda_{m}\mathcal{E}_{\eta,N_{B}%
}^{\otimes n}(\rho_{m})\}\leq2^{-n(R-g\left(  \eta N_{S}/[(1-\eta
)N_{B}+1]\right)  -\delta_{2}-\delta_{3})}+2\sqrt{\delta_{1}(n)+\delta
_{4}(n)+2\sqrt{\delta_{1}(n)}},
\]
where $\mathcal{E}_{\eta,N_{B}}^{\otimes n}$ denotes $n$ instances of the
thermal channel, $\delta_{1}(n)$ is defined in (\ref{Constraint}), $\delta
_{2},\delta_{3}$ are arbitrarily small positive constants, and $\delta_{4}(n)$
is a function decreasing to zero as $n$ increases. Thus, for any rate
$R>g\left(  \eta N_{S}/[(1-\eta)N_{B}+1]\right)  +\delta_{2}+\delta_{3}$,
(note that we can pick $\delta_{2}$ and $\delta_{3}$ small enough) the success
probability of any family of codes satisfying (\ref{Constraint}) decreases to
zero in the limit of large $n$.
\end{theorem}

\begin{proof}
We can consider the thermal noise channel $\mathcal{E}_{\eta,N_{B}}$ as a
cascade of a pure-loss bosonic channel $\mathcal{E}_{\eta_{1},0}$ followed by
an amplifier channel $\mathcal{A}_{G_{1}}$ (recall the discussion in
Section~\ref{sec:decompositions}):%
\begin{equation}
\mathcal{E}_{\eta,N_{B}}(\rho)=\left(  \mathcal{A}_{G_{1}}\circ\mathcal{E}%
_{\eta_{1},0}\right)  (\rho), \label{compositionII}%
\end{equation}
where the gain of the amplifier channel is $G_{1}=(1-\eta)N_{B}+1$ and the
pure-loss bosonic channel has transmissivity ${\eta}_{1}=\eta/G_{1}$. Let
$N_{S}^{\prime}\equiv N_{S}/[(1-\eta)N_{B}+1]$.

The average success probability of correctly decoding any code satisfying
(\ref{Constraint}) can then be written with the above decomposition rule and 
bounded as%
\begin{align}
&  \frac{1}{M}\sum_{m}\text{Tr}\{\Lambda_{m}\mathcal{E}_{\eta,N_{B}}^{\otimes
n}(\rho_{m})\}\nonumber\\
&  =\frac{1}{M}\sum_{m}\text{Tr}\{\Lambda_{m}(\mathcal{A}_{{G}_{1}}^{\otimes
n}\circ\mathcal{E}_{\eta_{1},0}^{\otimes n})(\rho_{m})\}\nonumber\\
&  =\frac{1}{M}\sum_{m}\text{Tr}\left\{  \left(  \mathcal{A}_{{G}_{1}%
}^{\otimes n}\right)  ^{\dag}\left(  \Lambda_{m}\right)  \mathcal{E}_{\eta
_{1},0}^{\otimes n}(\rho_{m})\right\} \nonumber\\
&  \leq\frac{1}{M}\sum_{m}\text{Tr}\left\{  \left(  \mathcal{A}_{{G}_{1}%
}^{\otimes n}\right)  ^{\dag}\left(  \Lambda_{m}\right)  \Pi_{\lceil n(\eta
N_{s}^{\prime}+\delta_{2})\rceil}\mathcal{E}_{\eta_{1},0}^{\otimes n}(\rho
_{m})\Pi_{\lceil n(\eta N_{s}^{\prime}+\delta_{2})\rceil}\right\} \nonumber\\
&  \ \ \ \ \ \ +\frac{1}{M}\sum_{m}\left\Vert\mathcal{E}_{\eta_{1},0}^{\otimes n}%
(\rho_{m})-\Pi_{\lceil n(\eta N_{s}^{\prime}+\delta_{2})\rceil}\mathcal{E}%
_{\eta_{1},0}^{\otimes n}(\rho_{m})\Pi_{\lceil n(\eta N_{s}^{\prime}%
+\delta_{2})\rceil}\right\Vert_{1}\nonumber\\
&  \leq\frac{1}{M}\sum_{m}\text{Tr}\left\{  \left(  \mathcal{A}_{{G}_{1}%
}^{\otimes n}\right)  ^{\dag}\left(  \Lambda_{m}\right)  (\Pi_{\lceil n(\eta
N_{s}^{\prime}+\delta_{2})\rceil}\mathcal{E}_{\eta_{1},0}^{\otimes n}(\rho
_{m})\Pi_{\lceil n(\eta N_{s}^{\prime}+\delta_{2})\rceil})\right\} \nonumber\\
&  \ \ \ \ \ \ +2\sqrt{\delta_{1}(n)+\delta_{4}(n)+2\sqrt{\delta_{1}(n)}}.
\end{align}
The first equality is obtained by using the decomposition rule stated above.
The second equality follows by defining $\left(  \mathcal{A}_{{G}_{1}%
}^{\otimes n}\right)  ^{\dag}$ as the adjoint of $\mathcal{A}_{{G}_{1}%
}^{\otimes n}$. The first inequality is a special case of the inequality
\begin{equation}
\operatorname{Tr}\{\Lambda\sigma\}\leq\text{Tr}\{\Lambda\rho\}+||\rho
-\sigma||_{1}, \label{ineq1}%
\end{equation}
which holds for $0\leq\Lambda\leq I$, $\rho,\sigma\geq0$, and
$\operatorname{Tr}\{\rho\},\operatorname{Tr}\{\sigma\}\leq1$. The second
inequality follows from a variation of the Gentle Measurement
Lemma~\cite{ON07,W99} for ensembles, which states that $\sum_{x}p_{X}%
(x)||\rho_{x}-\sqrt{\Lambda}\rho_{x}\sqrt{\Lambda}||_{1}\leq2\sqrt{\epsilon}$
for an ensemble $\{p_{X}(x),\rho_{x}\}$ where $\sum_{x}p_{X}\operatorname{Tr}%
\{\Lambda\rho_{x}\}\geq1-\varepsilon$ and $0\leq\varepsilon\leq1$. It also follows
from an application of Lemma~4 of \cite{StrongConversePureLoss}, with
$\delta_4(n)$ chosen as given there.

We now focus on the first term in the above expression to obtain the upper
bound in the statement of the theorem:%
\begin{align}
&  \frac{1}{M}\sum_{m}\operatorname{Tr}\{\left(  \mathcal{A}_{{G}_{1}%
}^{\otimes n}\right)  ^{\dag}\left(  \Lambda_{m}\right)  (\Pi_{\lceil n(\eta
N_{s}^{\prime}+\delta_{2})\rceil}\mathcal{E}_{\eta_{1},0}^{\otimes n}(\rho
_{m})\Pi_{\lceil n(\eta N_{s}^{\prime}+\delta_{2})\rceil})\}\\
&  \leq\frac{1}{M}\sum_{m}\operatorname{Tr}\{\Pi_{\lceil n(\eta N_{s}^{\prime
}+\delta_{2})\rceil}\left(  \mathcal{A}_{{G}_{1}}^{\otimes n}\right)  ^{\dag
}(\Lambda_{m})\Pi_{\lceil n(\eta N_{s}^{\prime}+\delta_{2})\rceil
}\}\nonumber\\
&  =M^{-1}\operatorname{Tr}\{\Pi_{\lceil n(\eta N_{s}^{\prime}+\delta
_{2})\rceil}\}\nonumber\\
&  \leq2^{-n(R-g(\eta N_{S}^{\prime})-\delta_{2}-\delta_{3})}.\nonumber
\end{align}
The first inequality follows since $||\mathcal{E}_{\eta_{1},0}^{\otimes
n}(\rho_{m})||_{\infty}\leq1$. The first equality is a consequence of the fact
that $\sum_{m}\Lambda_{m}=I$ for any POVM $\{\Lambda_{m}\}$ and that the
adjoint of any CPTP map is unital. The last equality
follows from (\ref{Rank}) and from the fact that the rate $R=(\log_{2}M)/n$.
Thus, we arrive at the statement of the theorem---if the rate $R>g(\eta
N_{s}^{\prime})$, we can choose the constants $\delta_{2},\delta_{3}$ to be
arbitrarily small such that $R>g(\eta N_{s}^{\prime})+\delta_{2}+\delta_{3}$,
and the success probability decreases to zero in the limit of $n\rightarrow
\infty$.
\end{proof}

\subsection{Koenig-Smith-like bound is a strong converse rate for the additive
noise channel}

Recall that the additive noise channel $\mathcal{N}_{\bar{n}}$ can be realized
as a pure-loss bosonic channel $\mathcal{E}_{\eta_{2},0}$ with transmissivity
$\eta_{2}=1/(\bar{n}+1)$ followed by an amplifier channel $\mathcal{A}_{G_{2}%
}$ with gain $G_{2}=(\bar{n}+1)$ (see Appendix~{\ref{symplectic}} for
details), i.e.,%
\begin{equation}
\mathcal{N}_{\bar{n}}=\left(  \mathcal{A}_{G_{2}}\circ\mathcal{E}_{\eta_{2}%
,0}\right)  (\rho)\equiv\mathcal{A}_{G_{2}}\left(  \mathcal{E}_{\eta_{2}%
,0}(\rho)\right)  \label{additivecomposition}\, .
\end{equation}
Then it follows from Theorem~\ref{TheoremI}, by making the replacement
$(1-\eta)N_{B}\rightarrow\bar{n}$ in the thermal noise channel results, that
the upper bound in (\ref{additiveUB1}) is a strong converse rate for the
additive noise channel. That is, the average success probability of correctly
decoding any code under a maximum photon number constraint decreases to zero
as $n$ becomes large for any rate $R>g\left(  N_{S}/(\bar{n}+1)\right)  $.

\subsection{Giovannetti \textit{et al}.~bound is a strong converse rate for
the thermal channel}

We now prove that the upper bound in (\ref{upperboundGV}) corresponds to a
strong converse rate for the thermal channel under a maximum photon number
constraint. In order to prove that, it is essential to show that if most of
the probability mass of the input state is in a subspace with photon number no
larger than $nN_{S}$, then the most of the probability mass of the thermal
channel output is in a subspace with photon number no larger than $n(\eta
N_{S}+\left(  1-\eta\right)  N_{B})$.

\begin{lemma}
\label{lem:photon-num-lemma}Let $\rho^{(n)}$ denote a density operator on $n$
modes satisfying%
\[
\operatorname{Tr}\{\Pi_{\lceil nN_{S}\rceil}\rho^{(n)}\}\geq1-\delta_{1}(n),
\]
where $\delta_{1}(n)$ is a function of $n$ decreasing to zero as $n$
increases. Then%
\[
\operatorname{Tr}\{\Pi_{\lceil n(\eta N_{S}+(1-\eta)N_{B})+\delta_{5})\rceil
}\mathcal{E}_{\eta,N_{B}}^{\otimes n}(\rho^{(n)})\}\geq1-\delta_{1}%
(n)-2\sqrt{\delta_{1}(n)}-\delta_{6}(n),
\]
where $\mathcal{E}_{\eta,N_{B}}^{\otimes n}$ represents $n$ instances of the
thermal noise channel, $\delta_{5}$ is an arbitrarily small positive constant,
and $\delta_{6}(n)$ is a function of $n$ decreasing to zero as $n\rightarrow
\infty$.
\end{lemma}

\begin{proof}
Recall the structural decomposition of the thermal noise channel from
(\ref{compositionI}):%
\[
\mathcal{E}_{\eta,N_{B}}(\rho)=\left(  \mathcal{N}_{(1-\eta)N_{B}}%
\circ\mathcal{E}_{\eta,0}\right)  (\rho).
\]
This decomposition states that a thermal noise channel with transmissivity
$\eta$\ and noise power $N_{B}$ can be realized as a concatenation of a
pure-loss channel of transmissivity $\eta$ followed by a classical noise
channel $\mathcal{N}_{\left(  1-\eta\right)  N_{B}}$. Thus, a photon number
state $\left\vert k\right\rangle \left\langle k\right\vert $\ input to the
thermal noise channel leads to an output of the following form:%
\begin{equation}
\mathcal{E}_{\eta,N_{B}}\left(  \left\vert k\right\rangle \left\langle
k\right\vert \right)  =\sum_{m=0}^{k}p_{m}\mathcal{N}_{\left(  1-\eta\right)
N_{B}}\left(  \left\vert m\right\rangle \left\langle m\right\vert \right)  ,
\label{classicaltothermal}%
\end{equation}
where%
\[
p_{m}=\binom{k}{m}\eta^{m}\left(  1-\eta\right)  ^{k-m}.
\]
The classical noise channel has the following action on a photon number
state~\cite{GGLMS04,Caves}:%
\[
\mathcal{N}_{\left(  1-\eta\right)  N_{B}}\left(  \left\vert m\right\rangle
\left\langle m\right\vert \right)  =\sum_{l=0}^{\infty}\lambda_{l}\left\vert
l\right\rangle \left\langle l\right\vert ,
\]
where%
\begin{equation}
\lambda_{l}=\sum_{j=0}^{\min\left(  l,m\right)  }\binom{l}{j}\binom{m}{j}%
\frac{\left(  \left(  1-\eta\right)  N_{B}\right)  ^{m+l-2j}}{\left(
1+\left(  1-\eta\right)  N_{B}\right)  ^{m+l+1}}. \label{eq:special-dist}%
\end{equation}
Important properties of the distribution $\lambda_{l}$ are that it decays
exponentially to zero as $l\rightarrow\infty$ and has finite second moment. It
follows from (\ref{classicaltothermal}) that%
\[
\mathcal{E}_{\eta,N_{B}}\left(  \left\vert k\right\rangle \left\langle
k\right\vert \right)  =\sum_{l=0}^{\infty}\left[  \sum_{m=0}^{k}\sum
_{j=0}^{\min\left(  l,m\right)  }\binom{k}{m}\eta^{m}\left(  1-\eta\right)
^{k-m}\binom{l}{j}\binom{m}{j}\frac{\left(  \left(  1-\eta\right)
N_{B}\right)  ^{m+l-2j}}{\left(  1+\left(  1-\eta\right)  N_{B}\right)
^{m+l+1}}\right]  \left\vert l\right\rangle \left\langle l\right\vert .
\]
The eigenvalues above represent a distribution over photon number states at
the output of the thermal noise channel $\mathcal{E}_{\eta,N_{B}}$, which we
can write as a probability distribution over $l$ given the input~$k$:%
\begin{equation}
p\left(  l|k\right)  =\sum_{m=0}^{k}\sum_{j=0}^{\min\left(  l,m\right)
}\binom{k}{m}\eta^{m}\left(  1-\eta\right)  ^{k-m}\binom{l}{j}\binom{m}%
{j}\frac{\left(  \left(  1-\eta\right)  N_{B}\right)  ^{m+l-2j}}{\left(
1+\left(  1-\eta\right)  N_{B}\right)  ^{m+l+1}}. \label{eq:conditional-dist}%
\end{equation}
The above probability distribution has its mean equal to $\eta k+\left(
1-\eta\right)  N_{B}$. The reason is that the mean photon number of the input
is equal to $k$, and the mean photon number of the output is equal to a linear
combination of the input mean photon numbers. Furthermore, it inherits from
the distribution in (\ref{eq:special-dist}) the properties of having finite
second moment and an exponential decay to zero as $l\rightarrow\infty$.

Supposing that the input state satisfies the maximum photon-number constraint
in (\ref{Constraint}), we now observe that%
\begin{align}
&  \text{Tr}\left\{  \left(  \Pi_{\left\lceil n\left(  \eta N_{S}+\left(
1-\eta\right)  N_{B}+\delta_{5}\right)  \right\rceil }\right)  \mathcal{E}%
_{\eta,N_{B}}^{\otimes n}\left(  \rho^{\left(  n\right)  }\right)  \right\}
\nonumber\\
&  \geq\text{Tr}\left\{  \left(  \Pi_{\left\lceil n\left(  \eta N_{S}+\left(
1-\eta\right)  N_{B}+\delta_{5}\right)  \right\rceil }\right)  \mathcal{E}%
_{\eta,N_{B}}^{\otimes n}\left(  \Pi_{\left\lceil nN_{S}\right\rceil }%
\rho^{\left(  n\right)  }\Pi_{\left\lceil nN_{S}\right\rceil }\right)
\right\}  -2\sqrt{\delta_{1}} \label{eq:photon-counting-output}%
\end{align}

The inequality follows from the Gentle Measurement Lemma~\cite{ON07,W99}.
Since there is photodetection at the output (i.e., the projector
$\Pi_{\left\lceil n\left(  \eta N_{S}+\left(  1-\eta\right)  N_{B}+\delta
_{5}\right)  \right\rceil }$ is diagonal in the number basis), it suffices for
us to consider the input $\Pi_{\left\lceil nN_{S}\right\rceil }\rho^{\left(
n\right)  }\Pi_{\left\lceil nN_{S}\right\rceil }$ to be diagonal in the
photon-number basis, and we write this as%
\[
\rho^{\left(  n\right)  }=\sum_{a^{n}:\sum_{i}a_{i}\leq\lceil nN_{S}\rceil
}p\left(  a^{n}\right)  \left\vert a^{n}\right\rangle \left\langle
a^{n}\right\vert ,
\]
where $\left\vert a^{n}\right\rangle $ represents strings of photon number
states. Continuing, we find that (\ref{eq:photon-counting-output}) is equal to%
\begin{align}
&  \sum_{a^{n}:\sum_{i}a_{i}\leq\left\lceil nN_{S}\right\rceil }p\left(
a^{n}\right)  \text{Tr}\left\{  \left(  \Pi_{\left\lceil n\left(  \eta
N_{S}+\left(  1-\eta\right)  N_{B}+\delta_{5}\right)  \right\rceil }\right)
\mathcal{E}_{\eta,N_{B}}^{\otimes n}\left(  \left\vert a^{n}\right\rangle
\left\langle a^{n}\right\vert \right)  \right\}  -2\sqrt{\delta_{1}%
}\nonumber\\
&  =\sum_{a^{n}:\sum_{i}a_{i}\leq\left\lceil nN_{S}\right\rceil }p\left(
a^{n}\right)  \sum_{l^{n}:\sum_{i}l_{i}\leq\left\lceil n\left(  \eta
N_{S}+\left(  1-\eta\right)  N_{B}+\delta_{5}\right)  \right\rceil }p\left(
l^{n}|a^{n}\right)  -2\sqrt{\delta_{1}}, \label{analysisterm}%
\end{align}
where the distribution $p\left(  l^{n}|a^{n}\right)  \equiv\prod
\limits_{i=1}^{n}p\left(  l_{i}|a_{i}\right)  $ and each $p\left(  l_{i}%
|a_{i}\right)  $ is defined from (\ref{eq:conditional-dist}).

In order to obtain a lower bound on the expression in (\ref{analysisterm}), we
analyze the term
\begin{equation}
\sum_{l^{n}:\sum_{i}l_{i}\leq\left\lceil n\left(  \eta N_{S}+\left(
1-\eta\right)  N_{B}+\delta_{5}\right)  \right\rceil }p\left(  l^{n}%
|a^{n}\right)  \label{conditional}%
\end{equation}
on its own under the assumption that $\sum_{i}a_{i}\leq\left\lceil
nN_{S}\right\rceil $. Let $L_{i}|a_{i}$ denote a conditional random variable
with distribution $p\left(  l_{i}|a_{i}\right)  $, and let $\overline{L^{n}%
}|a^{n}$ denote the sum random variable:%
\[
\overline{L^{n}}|a^{n}\equiv\sum_{i}L_{i}|a_{i},
\]
so that%
\begin{align}
\sum_{l^{n}:\sum_{i}l_{i}\leq\left\lceil n\left(  \eta N_{S}+\left(
1-\eta\right)  N_{B}+\delta_{5}\right)  \right\rceil }p\left(  l^{n}%
|a^{n}\right)   &  =\Pr\left\{  \overline{L^{n}}|a^{n}\leq n\left(  \eta
N_{S}+\left(  1-\eta\right)  N_{B}+\delta_{5}\right)  \right\} \nonumber\\
&  \geq\Pr\left\{  \overline{L^{n}}|a^{n}\leq n\left(  \eta\frac{1}{n}\sum
_{i}a_{i}+\left(  1-\eta\right)  N_{B}+\delta_{5}\right)  \right\}  ,
\label{eq:probability-to-bound}%
\end{align}
where the inequality follows from the constraint $\sum_{i}a_{i}\leq\left\lceil
nN_{S}\right\rceil $. Since%
\[
\mathbb{E}\left\{  L_{i}|a_{i}\right\}  =\eta a_{i}+\left(  1-\eta\right)
N_{B},
\]
it follows that%
\[
\mathbb{E}\left\{  \overline{L^{n}}|a^{n}\right\}  =n\left(  \eta\frac{1}%
{n}\sum_{i}a_{i}+\left(  1-\eta\right)  N_{B}\right)  ,
\]
and so the expression in (\ref{eq:probability-to-bound}) is the probability
that a sum of independent random variables deviates from its mean by no more
than $\delta_{5}$.

There are several ways to proceed with bounding the probability in
(\ref{eq:probability-to-bound}) from below. Since all of the random variables
$L_{i}|a_{i}$ have finite second moment, we can employ the Chebyshev
inequality to bound (\ref{eq:probability-to-bound}) from below by $1-C/n$,
where $C$ is a constant depending on the maximum variance of the $L_{i}$
random variables and the deviation $\delta_{5}$. However, if we would like to
prove that there is a stronger rate of convergence, the fact that the random
variables are unbounded might seem to be problematic. Nevertheless, one could
employ the truncation method detailed in Section~2.1 of \cite{T12}, in which
each random variable $L_{i}|a_{i}$ is split into two parts:%
\begin{align*}
\left(  L_{i}|a_{i}\right)  _{>T}  &  \equiv\left(  L_{i}|a_{i}\right)
\mathcal{I}\left(  \left(  L_{i}|a_{i}\right)  >T\right)  ,\\
\left(  L_{i}|a_{i}\right)  _{\leq T}  &  \equiv\left(  L_{i}|a_{i}\right)
\mathcal{I}\left(  \left(  L_{i}|a_{i}\right)  \leq T\right)  ,
\end{align*}
where $\mathcal{I}\left(  \cdot\right)  $ is the indicator function and $T$ is
a truncation parameter taken to be very large (much larger than $\max_{i}%
a_{i}$, for example). We can then split the sum random variable into two parts
as well:%
\begin{align*}
\overline{L^{n}}|a^{n}  &  =\left(  \overline{L^{n}}|a^{n}\right)
_{>T}+\left(  \overline{L^{n}}|a^{n}\right)  _{\leq T}\\
&  \equiv\sum_{i}\left(  L_{i}|a_{i}\right)  _{>T}+\sum_{i}\left(  L_{i}%
|a_{i}\right)  _{\leq T}.
\end{align*}
We can use the union bound to argue that%
\begin{multline}
\Pr\left\{  \overline{L^{n}}|a^{n}\geq\mathbb{E}\left\{  \overline{L^{n}%
}|a^{n}\right\}  +n\delta_{5}\right\}     \leq\Pr\left\{  \left(
\overline{L^{n}}|a^{n}\right)  _{>T}\geq\mathbb{E}\left\{  \left(
\overline{L^{n}}|a^{n}\right)  _{>T}\right\}  +n\delta_{5}/2\right\} \\
  +\Pr\left\{  \left(  \overline{L^{n}}|a^{n}\right)  _{\leq T}\geq
\mathbb{E}\left\{  \left(  \overline{L^{n}}|a^{n}\right)  _{\leq T}\right\}
+n\delta_{5}/2\right\}  .
\end{multline}
The idea from here is that for a random variable $L_{i}|a_{i}$ with sufficient
decay for large values, we can bound the first probability for $\left(
\overline{L^{n}}|a^{n}\right)  _{>T}$ from above by $\varepsilon/\delta_{5}%
$\ for $\varepsilon$ an arbitrarily small positive constant (made small by
taking $T$ larger)\ by employing the Markov inequality. We then bound the
second probability for $\left(  \overline{L^{n}}|a^{n}\right)  _{\leq T}$
using a Chernoff bound, since these random variables are bounded. This latter
bound has an exponential decay due to the ability to use a Chernoff bound. So,
the argument is just to make $\varepsilon$ arbitrarily small by increasing the
truncation parameter $T$, and for $n$ large enough, exponential convergence to
zero kicks in. We point the reader to Section~2.1 of \cite{T12} for more
details. By using either approach, we arrive at the following bound:%
\[
\sum_{l^{n}:\sum_{i}l_{i}\leq\left\lceil n\left(  \eta N_{S}+\left(
1-\eta\right)  N_{B}+\delta_{5}\right)  \right\rceil }p\left(  l^{n}%
|a^{n}\right)  \geq1-\delta_{6}(n),
\]
where $\delta_{6}(n)$ is a function decreasing to zero as $n\rightarrow\infty$.

Finally, we put this together with (\ref{analysisterm}) to obtain that%
\begin{align*}
&  \operatorname{Tr}\{\Pi_{\lceil n(\eta N_{S}+(1-\eta)N_{B})+\delta
_{5})\rceil}\mathcal{E}_{\eta,N_{B}}^{\otimes n}(\rho^{(n)})\}\\
&  \geq\sum_{a^{n}:\sum_{i}a_{i}\leq\left\lceil nN_{S}\right\rceil }p\left(
a^{n}\right)  \sum_{l^{n}:\sum_{i}l_{i}\leq\left\lceil n\left(  \eta
N_{S}+\left(  1-\eta\right)  N_{B}+\delta_{5}\right)  \right\rceil }p\left(
l^{n}|a^{n}\right)  -2\sqrt{\delta_{1}}\\
&  \geq\left(  1-\delta_{1}\right)  \left(  1-\delta_{6}(n)\right)
-2\sqrt{\delta_{1}}\\
&  \geq1-\delta_{1}-\delta_{6}(n)-2\sqrt{\delta_{1}}.
\end{align*}
\end{proof}

The above lemma can be extended to the additive noise channel by employing the
relation of this channel to the thermal channel (discussed in
Section~\ref{sec:decompositions}). For the additive noise channel
$\mathcal{N}_{\bar{n}}$, it follows that if the input state is in the subspace
with photon number no larger than $nN_{S}$, then the additive noise output is
projected with very high probability onto a subspace with photon number no
larger than $n(N_{S}+\bar{n})$.

We now proceed to prove that the upper bounds in (\ref{upperboundGV}) and
(\ref{additiveUB2})\ are strong converse rates.

\begin{theorem}
[(\ref{upperboundGV}) is a strong converse rate]\label{TheoremII} The average
success probability of any code satisfying (\ref{Constraint}) is bounded as
follows:%
\begin{multline*}
\frac{1}{M}\sum_{m}\text{Tr}\{\Lambda_{m}\mathcal{E}_{\eta,N_{B}}^{\otimes
n}(\rho_{m})\}\\
\leq2^{-n\left[  R-\left[  g\left(  \eta N_{S}+(1-\eta)N_{B}+\delta
_{5}\right)  -\log_{2}\left(  1+2(1-\eta)N_{B}\right)  \right]  +\frac{1}%
{n}\log_{2}\left(  n\right)  -\delta\right]  }\\
+\frac{1}{n}+2\sqrt{\delta_{1}(n)+2\sqrt{\delta_{1}(n)}+\delta_{6}(n)}%
\end{multline*}
where $\mathcal{E}_{\eta,N_{B}}^{\otimes n}$ denotes $n$ instances of the
thermal channel, and $\delta$ is an arbitrarily small positive constant. Thus,
if $R>[g\left(  \eta N_{S}+\left(  1-\eta\right)  N_{B}\right)  -\log
_{2}\left(  1+2(1-\eta)N_{B}\right)  ]$, then we can pick $\delta$ and
$\delta_{5}$ decreasing to zero for large $n$, such that the success
probability of any family of codes satisfying (\ref{Constraint}) decreases to
zero as $n\rightarrow\infty$.
\end{theorem}

\begin{proof}
Consider the success probability of any code satisfying the maximum
photon-number constraint in (\ref{Constraint}):%
\[
\frac{1}{M}\sum_{m}\text{Tr}\{\Lambda_{m}\mathcal{E}_{\eta,N_{B}}^{\otimes
n}(\rho_{m})\}.
\]
From the assumption that%
\[
\frac{1}{M}\sum_{m}\text{Tr}\left\{  \Pi_{\lceil nN_{S}\rceil}\rho
_{m}\right\}  \geq1-\delta_{1}\left(  n\right)  ,
\]
Lemma~\ref{lem:photon-num-lemma} allows us to conclude that%
\[
\frac{1}{M}\sum_{m}\operatorname{Tr}\{\Pi_{\lceil n(\eta N_{S}+(1-\eta
)N_{B})+\delta_{5})\rceil}\mathcal{E}_{\eta,N_{B}}^{\otimes n}(\rho_{m}%
)\}\geq1-\delta_{1}(n)-2\sqrt{\delta_{1}(n)}-\delta_{6}(n),
\]
where the functions and constants are as given there. Using the Gentle
Measurement Lemma for ensembles \cite{ON07,W99}, we find that%
\begin{align*}
&  \frac{1}{M}\sum_{m}\text{Tr}\{\Lambda_{m}\mathcal{E}_{\eta,N_{B}}^{\otimes
n}(\rho_{m})\}\\
&  \leq\frac{1}{M}\sum_{m}\text{Tr}\{\Lambda_{m}\Pi_{\lceil n(\eta
N_{S}+(1-\eta)N_{B})+\delta_{5})\rceil}\mathcal{E}_{\eta,N_{B}}^{\otimes
n}(\rho_{m})\Pi_{\lceil n(\eta N_{S}+(1-\eta)N_{B})+\delta_{5})\rceil}\}\\
&  \ \ \ \ \ \ +\frac{1}{M}\sum_{m}\left\Vert \Pi_{\lceil n(\eta N_{S}%
+(1-\eta)N_{B})+\delta_{5})\rceil}\mathcal{E}_{\eta,N_{B}}^{\otimes n}%
(\rho_{m})\Pi_{\lceil n(\eta N_{S}+(1-\eta)N_{B})+\delta_{5})\rceil
}-\mathcal{E}_{\eta,N_{B}}^{\otimes n}(\rho_{m})\right\Vert _{1}\\
&  \leq\frac{1}{M}\sum_{m}\text{Tr}\{\Lambda_{m}\Pi_{\lceil n(\eta
N_{S}+(1-\eta)N_{B})+\delta_{5})\rceil}\mathcal{E}_{\eta,N_{B}}^{\otimes
n}(\rho_{m})\Pi_{\lceil n(\eta N_{S}+(1-\eta)N_{B})+\delta_{5})\rceil}\}\\
&  \ \ \ \ \ \ +2\sqrt{\delta_{1}(n)+2\sqrt{\delta_{1}(n)}+\delta_{6}(n)}%
\end{align*}
We now focus on the term%
\begin{align}
&  \frac{1}{M}\sum_{m}\text{Tr}\{\Lambda_{m}\Pi_{\lceil n(\eta N_{S}%
+(1-\eta)N_{B})+\delta_{5})\rceil}\mathcal{E}_{\eta,N_{B}}^{\otimes n}%
(\rho_{m})\Pi_{\lceil n(\eta N_{S}+(1-\eta)N_{B})+\delta_{5})\rceil
}\}\nonumber\\
&  =\frac{1}{M}\sum_{m}\text{Tr}\{\Pi_{\lceil n(\eta N_{S}+(1-\eta
)N_{B})+\delta_{5})\rceil}\Lambda_{m}\Pi_{\lceil n(\eta N_{S}+(1-\eta
)N_{B})+\delta_{5})\rceil}\mathcal{E}_{\eta,N_{B}}^{\otimes n}(\rho
_{m})\}.\label{eq:succ-prob-projected}%
\end{align}
Rather than working with the states $\mathcal{E}_{\eta,N_{B}}^{\otimes
n}\left(  \rho_{m}\right)  $ directly, we consider states $\widetilde{\sigma}_{m}$ that
are $\varepsilon$-close in trace distance to $\mathcal{E}_{\eta,N_{B}}^{\otimes
n}\left(  \rho_{m}\right)  $, and this gives the following upper bound on
(\ref{eq:succ-prob-projected}):%
\begin{align*}
&  \leq\frac{1}{M}\sum_{m}\text{Tr}\{\Pi_{\lceil n(\eta N_{S}+(1-\eta
)N_{B})+\delta_{5})\rceil}\Lambda_{m}\Pi_{\lceil n(\eta N_{S}+(1-\eta
)N_{B})+\delta_{5})\rceil}\widetilde{\sigma}_{m}\}+\varepsilon\\
&  \leq\frac{1}{M}\sum_{m}\text{Tr}\{\Pi_{\lceil n(\eta N_{S}+(1-\eta
)N_{B})+\delta_{5})\rceil}\Lambda_{m}\Pi_{\lceil n(\eta N_{S}+(1-\eta
)N_{B})+\delta_{5})\rceil}\}\left\Vert \widetilde{\sigma}_{m}\right\Vert
_{\infty}+\varepsilon.
\end{align*}
Now, this last bound holds regardless of which $\widetilde{\sigma}_{m}$ we
pick, so we optimize over all of them that are $\varepsilon$-close to
$\mathcal{E}_{\eta,N_{B}}^{\otimes n}\left(  \rho_{m}\right)  $ (let us denote
this set by $\mathcal{B}^{\varepsilon}\left(  \mathcal{E}_{\eta,N_{B}%
}^{\otimes n}\left(  \rho_{m}\right)  \right)  $). This gives the tightest
upper bound on the success probability, leading to%
\begin{align*}
&  \frac{1}{M}\sum_{m}\text{Tr}\{\Pi_{\lceil n(\eta N_{S}+(1-\eta
)N_{B})+\delta_{5})\rceil}\Lambda_{m}\Pi_{\lceil n(\eta N_{S}+(1-\eta
)N_{B})+\delta_{5})\rceil}\mathcal{E}_{\eta,N_{B}}^{\otimes n}(\rho_{m})\}\\
&  \leq\frac{1}{M}\sum_{m}\text{Tr}\{\Pi_{\lceil n(\eta N_{S}+(1-\eta
)N_{B})+\delta_{5})\rceil}\Lambda_{m}\Pi_{\lceil n(\eta N_{S}+(1-\eta
)N_{B})+\delta_{5})\rceil}\}\inf_{\widetilde{\sigma}_{m}\in\mathcal{B}%
^{\varepsilon}\left(  \mathcal{E}_{\eta,N_{B}}^{\otimes n}\left(  \rho
_{m}\right)  \right)  }\left\Vert \widetilde{\sigma}_{m}\right\Vert _{\infty
}+\varepsilon%
\end{align*}
The quantity $\inf_{\widetilde{\sigma}_{m}\in\mathcal{B}^{\varepsilon}\left(
\mathcal{E}_{\eta,N_{B}}^{\otimes n}\left(  \rho_{m}\right)  \right)
}\left\Vert \widetilde{\sigma}_{m}\right\Vert _{\infty}$ is related to the
smooth min-entropy via%
\[
\inf_{\widetilde{\sigma}_{m}\in\mathcal{B}^{\varepsilon}\left(  \mathcal{E}%
_{\eta,N_{B}}^{\otimes n}\left(  \rho_{m}\right)  \right)  }\left\Vert
\widetilde{\sigma}_{m}\right\Vert _{\infty}=2^{-H_{\min}^{\varepsilon}\left(
\mathcal{E}_{\eta,N_{B}}^{\otimes n}(\rho_{m})\right)  },
\]
so we replace the expression above by%
\begin{align*}
&  \frac{1}{M}\sum_{m}\text{Tr}\{\Pi_{\lceil n(\eta N_{S}+(1-\eta
)N_{B})+\delta_{5})\rceil}\Lambda_{m}\Pi_{\lceil n(\eta N_{S}+(1-\eta
)N_{B})+\delta_{5})\rceil}\}2^{-H_{\min}^{\varepsilon}\left(  \mathcal{E}%
_{\eta,N_{B}}^{\otimes n}(\rho_{m})\right)  }+\varepsilon \\
&  \leq\frac{1}{M}\sum_{m}\text{Tr}\{\Pi_{\lceil n(\eta N_{S}+(1-\eta
)N_{B})+\delta_{5})\rceil}\Lambda_{m}\Pi_{\lceil n(\eta N_{S}+(1-\eta
)N_{B})+\delta_{5})\rceil}\}\sup_{\rho}2^{-H_{\min}^{\varepsilon}\left(
\mathcal{E}_{\eta,N_{B}}^{\otimes n}(\rho)\right)  }+\varepsilon \\
&  =\frac{1}{M}2^{-\inf_{\rho}H_{\min}^{\varepsilon}\left(  \mathcal{E}%
_{\eta,N_{B}}^{\otimes n}(\rho)\right)  }\text{Tr}\{\Pi_{\lceil n(\eta
N_{S}+(1-\eta)N_{B})+\delta_{5})\rceil}\}+\varepsilon \\
&  \leq\frac{1}{M}2^{-\inf_{\rho}H_{\min}^{\varepsilon}\left(  \mathcal{E}%
_{\eta,N_{B}}^{\otimes n}(\rho)\right)  }2^{n\left[  g\left(  \eta
N_{S}+(1-\eta)N_{B}+\delta_{5}\right)  +\delta\right]  }+\varepsilon \\
&  \leq2^{-nR}2^{-\inf_{\rho}H_{2}\left(  \mathcal{E}_{\eta,N_{B}}^{\otimes
n}(\rho)\right)  +\log_{2}\left(  \frac{1}{\varepsilon}\right)  }2^{n\left[
g\left(  \eta N_{S}+(1-\eta)N_{B}+\delta_{5}\right)  +\delta\right]  }%
+\varepsilon \\
&  =2^{-nR}2^{-n\left[  \log_{2}\left(  1+2(1-\eta)N_{B})\right)  +\frac{1}%
{n}\log_{2}\left(  \frac{1}{\varepsilon}\right)  \right]  }2^{n\left[
g\left(  \eta N_{S}+(1-\eta)N_{B}+\delta_{5}\right)  +\delta\right]  }%
+\varepsilon.
\end{align*}
The first inequality follows by taking a supremum over all input states. The
first equality follows because $\sum_{m}\Lambda_{m}=I$. The second inequality
follows from the upper bound in (\ref{Rank}) on the rank of the photon number
subspace projector. The last few lines follow by applying the following
relation from~\cite{RWISIT1} between the smooth min-entropy $H_{\min}^{\varepsilon}\left(
\mathcal{E}_{\eta,N_{B}}^{\otimes n}\left(  \rho_{m}\right)  \right)  $ and
the quantum R\'{e}nyi entropy $H_{\alpha}\left(  \mathcal{E}_{\eta,N_{B}%
}^{\otimes n}\left(  \rho_{m}\right)  \right)  $ in (\ref{Renyismoothing}%
) for $\alpha=2$, yielding%
\begin{align*}
\inf_{\rho}H_{\min}^{\varepsilon}\left(  \mathcal{E}_{\eta,N_{B}}^{\otimes
n}\left(  \rho_{m}\right)  \right)   &  \geq\inf_{\rho}H_{2}\left(
\mathcal{E}_{\eta,N_{B}}^{\otimes n}\left(  \rho\right)  \right)  -\log\left(
\frac{1}{\varepsilon}\right)  \\
&  \geq n\inf_{\omega}H_{2}\left(  \mathcal{E}_{\eta,N_{B}}\left(
\omega\right)  \right)  -\log\left(  \frac{1}{\varepsilon}\right)  \\
&  =n\log_{2}\left(  1+2(1-\eta)N_{B}\right)  -\log\left(  \frac
{1}{\varepsilon}\right)  ,
\end{align*}
where the last two lines above follow from the main result of
\cite{GGLMS04,GLMS04}, that the output R\'{e}nyi entropy of order two is
minimized by the $n$-fold tensor-product vacuum state. We now see that we can choose $\varepsilon
=\frac{1}{n}$, and we arrive at the following bound%
\[
2^{-nR}2^{-n\left[  \log_{2}\left(  1+2(1-\eta)N_{B}\right)  +\frac{1}{n}%
\log_{2}\left(  n\right)  \right]  }2^{n\left[  g\left(  \eta N_{S}%
+(1-\eta)N_{B}+\delta_{5}\right)  +\delta\right]  }+\frac{1}{n}.
\]
Putting everything together, we arrive at the following inequality:%
\begin{multline*}
\frac{1}{M}\sum_{m}\text{Tr}\{\Lambda_{m}\mathcal{E}_{\eta,N_{B}}^{\otimes
n}(\rho_{m})\}\\
\leq2^{-n\left[  R-\left[  g\left(  \eta N_{S}+(1-\eta)N_{B}+\delta
_{5}\right)  -\log_{2}\left(  1+2(1-\eta)N_{B}\right)  \right]  +\frac{1}%
{n}\log_{2}\left(  n\right)  -\delta\right]  }\\
+\frac{1}{n}+2\sqrt{\delta_{1}(n)+2\sqrt{\delta_{1}(n)}+\delta_{6}(n)}%
\end{multline*}

This upper bound on the success probability demonstrates that for a rate%
\[
R>\left[  g\left(  \eta N_{S}+\left(  1-\eta\right)  N_{B}\right)  -\log
_{2}\left(  1+2(1-\eta)N_{B}\right)  \right]  ,
\]
we can choose $\delta_{5}$ and $\delta$ small enough so that the success
probability decreases to zero in the limit of large $n$.
\end{proof}

\subsection{Giovannetti \textit{et al}.~bound is a strong converse rate for
the additive noise channel}

Using (\ref{additivecomposition}), it follows that we can take the limit
$(1-\eta)N_{B}\rightarrow\bar{n}$ (with $N_{B}\rightarrow\infty$ and
$\eta\rightarrow1$) to prove that the upper bound in (\ref{additiveUB2})
serves as a strong converse rate for the additive noise channel~$\mathcal{N}%
_{\bar{n}}$. The arguments for showing the strong converse rate for the
thermal channel then apply for the additive noise channel, and we can say
that, for the additive noise channel, the average success probability under a
maximum photon number constraint decreases to zero with many channel uses if
$R>[g\left(  N_{S}+\bar{n}\right)  -\log_{2}\left(  1+2\bar{n}\right)  ]$.

\section{Conclusion}

\label{sec:conclusion}

In this paper, we showed that several previously known upper bounds on the
classical capacity for the thermal noise and additive noise channels are
actually strong converse rates. We did this by imposing a particular maximum photon number
constraint, guaranteeing that the inputs to the channel have almost all of
their shadow on a subspace with photon number no larger than $nN_{S}$. The
classical capacity of these two channels are not exactly known; however, our
results strengthen the interpretation of known upper bounds on the classical
capacity of these two channels, so that there is no 
room for a trade-off between the communication rate and 
error probability. Besides having an application to proving security in
some particular models of cryptography~\cite{KS4}, it should be possible to
extend our results to multiple-access bosonic channels, i.e., bosonic channels
in which two or more senders communicate to a common receiver over a shared
channel~\cite{MAC}.

\textit{Note}: After the completion of this work, we discovered very recently that
Giovannetti, Holevo, and Garcia-Patron proposed a solution to the long-standing
minimum output entropy conjecture \cite{GHG13}\ (in fact a more general
Gaussian optimizer conjecture). Their results imply that the lower bounds in
(\ref{eq:LB}) and (\ref{eq:additiveLB}) are in fact upper bounds as well, so
that they have identified the classical capacity of these channels. After
browsing their proof, we think that it should be possible to combine their
results with the development here in order to prove that the rates in
(\ref{eq:LB}) and (\ref{eq:additiveLB}) are in fact strong converse rates (so
that there is a strong converse theorem for the classical capacity of these
channels). In order to arrive at this conclusion, one would need at the very
least to extend their development in Section~6\ to prove that the minimum
output R\'enyi entropy for all $\alpha\geq1$ is minimized by the vacuum state.
One could then take a similar approach as we did in the last few steps of the
proof of Theorem~\ref{TheoremII}. However, this remains the subject of future research.

\section{Acknowledgements}

We are grateful to Raul Garcia-Patron and Andreas Winter for insightful
discussions. BRB would like to acknowledge supports from an Army Research
Office grant (W911NF-13-1-0381) and the Charles E. Coates Memorial Research 
Award Grant from Louisiana State University. MMW is grateful to the Department of Physics
and Astronomy at Louisiana State University for startup funds that supported
this research.

\appendix

\section{Relation between smooth min-entropy and R\'enyi entropy}

\label{smoothentropy}

For completeness, we include a proof of the main result of \cite{RWISIT1}:

\begin{lemma}
For any random variable $Z$, $\alpha>1$, and $\varepsilon\in\left(
0,1\right)  $, the following inequality holds%
\[
H_{\min}^{\varepsilon}\left(  Z\right)  \geq H_{\alpha}\left(  Z\right)
-\frac{1}{\alpha-1}\log\left(  \frac{1}{\varepsilon}\right)  .
\]

\end{lemma}

\begin{proof}
Let $p_{Z}\left(  z\right)  $ be a probability distribution for $Z$. Suppose
without loss of generality that the elements of the distribution are in
decreasing order. In order to prove this inequality, we should find another
distribution $q_{Z}\left(  z\right)  $ such that $\frac{1}{2}\sum
_{z}\left\vert p_{Z}\left(  z\right)  -q_{Z}\left(  z\right)  \right\vert
=\varepsilon$ and for which the inequality holds. We will choose it to have
the same support as $p_{Z}\left(  z\right)  $. To this end, let $p$ be a real
less than $p_{\max}\left(  Z\right)  $ and such that%
\[
\sum_{z\in\mathcal{Z}_{p}}\left(  p_{Z}\left(  z\right)  -p\right)
=\varepsilon,
\]
where $\mathcal{Z}_{p}\equiv\left\{  z:p_{Z}\left(  z\right)  \geq p\right\}
$. (We assume here that $\varepsilon$ is small enough such that $p$ could be
the maximum probability of a legitimate probability distribution.) Then this
relation implies that%
\begin{equation}
\sum_{z\in\mathcal{Z}_{p}}p_{Z}\left(  z\right)  =\left\vert \mathcal{Z}%
_{p}\right\vert p+\varepsilon. \label{eq:robin-hood}%
\end{equation}

Now we consider $q_{Z}\left(  z\right)  $ as a uniform redistribution of the
excess probability $\sum_{z\in\mathcal{Z}_{p}}\left(  p_{Z}\left(  z\right)
-p\right)  =\varepsilon$ to probabilities with values $z\notin\mathcal{Z}_{p}%
$:%
\[
q_{Z}\left(  z\right)  =\left\{
\begin{array}
[c]{cc}%
p & z\in\mathcal{Z}_{p}\\
p_{Z}\left(  z\right)  \left(  1+\frac{\varepsilon}{\sum_{z\notin
\mathcal{Z}_{p}}p_{Z}\left(  z\right)  }\right)  & z\notin\mathcal{Z}_{p}%
\end{array}
\right.  .
\]
The function $q_{Z}\left(  z\right)  $ defined above is indeed a probability
distribution because all of its elements are non-negative and%
\begin{align*}
\sum_{z}q_{Z}\left(  z\right)   &  =\sum_{z\in\mathcal{Z}_{p}}q_{Z}\left(
z\right)  +\sum_{z\notin\mathcal{Z}_{p}}q_{Z}\left(  z\right) \\
&  =\left\vert \mathcal{Z}_{p}\right\vert p+\sum_{z\notin\mathcal{Z}_{p}}%
p_{Z}\left(  z\right)  \left(  1+\frac{\varepsilon}{\sum_{z\notin
\mathcal{Z}_{p}}p_{Z}\left(  z\right)  }\right) \\
&  =\left\vert \mathcal{Z}_{p}\right\vert p+\sum_{z\notin\mathcal{Z}_{p}}%
p_{Z}\left(  z\right)  +\varepsilon\\
&  =1,
\end{align*}
where in the third line we used (\ref{eq:robin-hood}). Furthermore, the
variational distance between $p_{Z}$ and $q_{Z}$ is%
\begin{align*}
&  \frac{1}{2}\sum_{z}\left\vert p_{Z}\left(  z\right)  -q_{Z}\left(
z\right)  \right\vert \\
&  =\frac{1}{2}\sum_{z\in\mathcal{Z}_{p}}\left\vert p_{Z}\left(  z\right)
-q_{Z}\left(  z\right)  \right\vert +\frac{1}{2}\sum_{z\notin\mathcal{Z}_{p}%
}\left\vert p_{Z}\left(  z\right)  -q_{Z}\left(  z\right)  \right\vert \\
&  =\frac{1}{2}\sum_{z\in\mathcal{Z}_{p}}\left\vert p_{Z}\left(  z\right)
-p\right\vert +\frac{1}{2}\sum_{z\notin\mathcal{Z}_{p}}\left\vert p_{Z}\left(
z\right)  -\left(  p_{Z}\left(  z\right)  \left(  1+\frac{\varepsilon}%
{\sum_{z\notin\mathcal{Z}_{p}}p_{Z}\left(  z\right)  }\right)  \right)
\right\vert \\
&  =\frac{1}{2}\varepsilon+\frac{1}{2}\sum_{z\notin\mathcal{Z}_{p}}%
p_{Z}\left(  z\right)  \frac{\varepsilon}{\sum_{z\notin\mathcal{Z}_{p}}%
p_{Z}\left(  z\right)  }\\
&  =\varepsilon.
\end{align*}
Now consider that%
\begin{align*}
\sum_{z}p_{Z}\left(  z\right)  ^{\alpha}  &  \geq\sum_{z\in\mathcal{Z}_{p}%
}p_{Z}\left(  z\right)  ^{\alpha}\\
&  \geq p^{\alpha-1}\sum_{z\in\mathcal{Z}_{p}}p_{Z}\left(  z\right) \\
&  \geq p^{\alpha-1}\varepsilon,
\end{align*}
where the second inequality follows from the definition of $\mathcal{Z}_{p}$,
which implies that $1\geq\left(  p/p_{Z}\left(  z\right)  \right)  ^{\alpha
-1}$ whenever $z\in\mathcal{Z}_{p}$. From this, we see that%
\[
-\log p\geq\frac{1}{1-\alpha}\log\sum_{z}p_{Z}\left(  z\right)  ^{\alpha
}-\frac{1}{\alpha-1}\log\left(  \frac{1}{\varepsilon}\right)  .
\]
This inequality then implies the statement of the lemma because%
\begin{align*}
H_{\min}^{\varepsilon}\left(  Z\right)   &  \geq H_{\min}\left(  q_{Z}\right)
\\
&  =-\log p.
\end{align*}
\end{proof}

A generalization of the above proof to the quantum setting considering all
density operators $\widetilde{\rho}$ that are $\varepsilon$-close to density
operator $\rho$ for $\varepsilon>0$ gives the following relation between
quantum smooth min-entropy $H_{\min}^{\varepsilon}\left(  \rho\right)  $ and
the quantum R\'{e}nyi entropy $H_{\alpha}\left(  \rho\right)  $%
(\ref{Renyismoothing}):{%
\[
H_{\min}^{\varepsilon}\left(  \rho\right)  \geq H_{\alpha}\left(  \rho\right)
-\frac{1}{\alpha-1}\log\left(  \frac{1}{\varepsilon}\right)  .
\]
The same proof works for density operators that act on a separable Hilbert
space (since such density operators are diagonalized by a countable orthonormal basis),
which is the case for our considerations in this paper.}

\section{Structural decompositions of the bosonic channels using symplectic
formalism}

\label{symplectic} Here, for completeness, we review in detail an argument for
the structural decompositions of the bosonic channels using the symplectic
formalism~\cite{symplectic,WPGCRSL12} (however, note that these results were
well known much before the present paper). In this formalism, the action of a
Gaussian channel is characterized by two matrices $X$ and $Y$ which act as
follows on covariance matrix $\Gamma$
\begin{equation}
\Gamma\longrightarrow\Gamma^{\prime}=X\Gamma X^{T}+Y,
\end{equation}
where $X^{T}$ is the transpose of the matrix $X$. Such a map is called as the
symplectic map which applies to any Gaussian channel. Below we describe the
symplectic transformations for each of the channels $\mathcal{N}_{\bar{n}}$,
$\mathcal{E}_{\eta,0}$, $\mathcal{A}_{G}$, and $\mathcal{E}_{\eta,N_{B}}$:

\begin{itemize}
\item {The additive noise channel $\mathcal{N}_{\bar{n}}$ with variance
$\bar{n}$ is given by
\begin{equation}
\label{m1}X = \mathbb{I}~\text{and}~Y = 2 \bar{n}~\mathbb{I},
\end{equation}
where $\mathbb{I}$ represents the identity matrix.
}
\item {The pure-loss channel $\mathcal{E}_{\eta,0}$ with transmissivity
$\eta<1$ is given by
\begin{equation}
\label{m2}X = \sqrt{\eta}~\mathbb{I}~\text{and}~Y = (1-\eta)~\mathbb{I}.
\end{equation}
}
\item {The thermal noise channel $\mathcal{E}_{\eta,N_{B}}$ with
transmissivity $\eta<1$ and noise photon number $N_{B}$ is given by
\begin{equation}
\label{m3}X = \sqrt{\eta}~\mathbb{I}~\text{and}~Y = (1-\eta)(2N_{B}%
+1)~\mathbb{I}.
\end{equation}
}
\item {The amplifier channel $\mathcal{A}_{G}$ with gain $G >1$ is given by
\begin{equation}
\label{m4}X = \sqrt{G}~\mathbb{I}~\text{and}~Y = (G-1)~\mathbb{I}.
\end{equation}
}
\end{itemize}

We now show that the additive noise channel $\mathcal{N} _{\bar{n}}$ can be
regarded as a pure-loss bosonic channel $\mathcal{E}_{\eta,0}$ with
$\eta=1/(\bar{n}+1)$ followed by an amplifier channel $\mathcal{A}_{G}$ with
$G=(\bar{n}+1)$. To do so, we substitute
\begin{align*}
X_{1}  &  =\sqrt{1/(\bar{n}+1)} \mathbb{I},\\
Y_{1}  &  =(1-(1/(\bar{n}+1))~\mathbb{I},\\
X_{2}  &  = \sqrt{(\bar{n}+1)}~\mathbb{I},\\
Y_{2}  &  = \bar{n}~\mathbb{I}.
\end{align*}
in (\ref{m2}) and (\ref{m4}), where $(X_{1},Y_{1})$ and $(X_{2},Y_{2})$
correspond to the pure-loss bosonic channel $\mathcal{E}_{\eta,0}$ and the
amplifier channel $\mathcal{A}_{G}$, respectively. The covariance matrix
$\Gamma_{12}$ for the composite map $(\mathcal{A}_{\bar{n}+1} \circ
\mathcal{E}_{\frac{1}{\bar{n}+1},0})$ is then obtained as $\Gamma_{12}=X_{2}
(X_{1} \Gamma X_{1}^{T} + Y_{1}) X_{2}^{T} + Y_{2}=\Gamma~\mathbb{I}+ 2
\bar{n} ~\mathbb{I}$, which represents the additive noise channel $\mathcal{N}
_{\bar{n}}$ [(\ref{m1})]. Thus, we recover the decomposition in
(\ref{c3})
\[
\mathcal{N}_{\bar{n}}(\rho)=(\mathcal{A}_{\bar{n}+1} \circ\mathcal{E}
_{\frac{1}{\bar{n}+1},0})(\rho)\, .
\]

Following a similar approach we can find the other structural decompositions
in (\ref{compositionI}) and (\ref{eq:compositionII}):
\begin{align*}
\mathcal{E}_{\eta,N_{B}}(\rho) &  =\left(  \mathcal{N}_{(1-\eta)N_{B}}%
\circ\mathcal{E}_{\eta,0}\right)  (\rho),\\
\mathcal{E}_{\eta,N_{B}}(\rho) &  =\left(  \mathcal{A}_{G}\circ\mathcal{E}%
_{\eta,0}\right)  (\rho).
\end{align*}

\bibliographystyle{plain}
\bibliography{Ref_new}

\begin{thebibliography}{10}

\bibitem{Arimoto}
Suguru Arimoto.
\newblock On the converse to the coding theorem for discrete memoryless
  channels.
\newblock {\em IEEE Transactions on Information Theory}, 19:357--359, May 1973.

\bibitem{CGH06}
Filippo Caruso, Vittorio Giovannetti, and Alexander~S. Holevo.
\newblock One-mode bosonic {Gaussian} channels: A full weak-degradability
  classification.
\newblock {\em New Journal of Physics}, 8(12):310, 2006.
\newblock arXiv:quant-ph/0609013.

\bibitem{Caves}
Carlton~M. Caves.
\newblock Hidden variable model for continuous-variable teleportation.
\newblock August 2003.
\newblock info.phys.unm.edu/~caves/reports/cvteleportation.pdf.

\bibitem{symplectic}
Jens Eisert and Michael~M Wolf.
\newblock Gaussian quantum channels.
\newblock {\em Quantum Information with Continuous Variables of Atoms and
  Light}, pages 23--42, 2007.
\newblock arXiv:quant-ph/0505151.

\bibitem{GNLSC12}
Raul Garcia-Patron, Carlos Navarrete-Benlloch, Seth Lloyd, Jeffrey~H. Shapiro,
  and Nicolas~J. Cerf.
\newblock Majorization theory approach to the {Gaussian} channel minimum
  entropy conjecture.
\newblock {\em Physical Review Letters}, 108:110505, March 2012.
\newblock arXiv:1111.1986.

\bibitem{GK04}
Christopher Gerry and Peter Knight.
\newblock {\em Introductory Quantum Optics}.
\newblock Cambridge University Press, November 2004.

\bibitem{GGLMS04}
Vittorio Giovannetti, Saikat Guha, Seth Lloyd, Lorenzo Maccone, and Jeffrey~H.
  Shapiro.
\newblock Minimum output entropy of bosonic channels: A conjecture.
\newblock {\em Physical Review A}, 70:032315, September 2004.
\newblock arXiv:quant-ph/0404005.

\bibitem{SGGLMY04}
Vittorio Giovannetti, Saikat Guha, Seth Lloyd, Lorenzo Maccone, Jeffrey~H.
  Shapiro, and Brent~J. Yen.
\newblock Capacity of bosonic communications.
\newblock {\em AIP Conference Proceedings: QCMC 2004}, 734:21, July 2004.

\bibitem{GGLMSY04}
Vittorio Giovannetti, Saikat Guha, Seth Lloyd, Lorenzo Maccone, Jeffrey~H.
  Shapiro, and Horace~P. Yuen.
\newblock Classical capacity of the lossy bosonic channel: The exact solution.
\newblock {\em Physical Review Letters}, 92(2):027902, January 2004.
\newblock arXiv:quant-ph/0308012.

\bibitem{GHG13}
Vittorio Giovannetti, Alexander~S. Holevo, and Ra\'{u}l Garc\'{\i}a-Patr\'{o}n.
\newblock A solution of the {Gaussian} optimizer conjecture.
\newblock December 2013.
\newblock arXiv:1312.2251.

\bibitem{PracticalPurposes}
Vittorio Giovannetti, Seth Lloyd, Lorenzo Maccone, and Jeffrey~H. Shapiro.
\newblock Electromagnetic channel capacity for practical purposes.
\newblock {\em Nature Photonics}, 7(10):834--838, October 2013.
\newblock arXiv:1210.3300.

\bibitem{GLMS04}
Vittorio Giovannetti, Seth Lloyd, Lorenzo Maccone, Jeffrey~H. Shapiro, and
  Brent~J. Yen.
\newblock Minimum {R\'enyi} and {Wehrl} entropies at the output of bosonic
  channels.
\newblock {\em Physical Review A}, 70:022328, August 2004.
\newblock arXiv:quant-ph/0404037.

\bibitem{Hall94}
Michael J.~W. Hall.
\newblock Gaussian noise and quantum-optical communication.
\newblock {\em Physical Review A}, 50:3295--3303, October 1994.

\bibitem{HO93}
Michael J.~W. Hall and M.~J. O'Rourke.
\newblock Realistic performance of the maximum information channel.
\newblock {\em Quantum Optics: Journal of the European Optical Society Part B},
  5(3):161, June 1993.

\bibitem{Hol98}
Alexander~S. Holevo.
\newblock The capacity of the quantum channel with general signal states.
\newblock {\em IEEE Transactions on Information Theory}, 44:269--273, 1998.

\bibitem{HG12}
Alexander~S. Holevo and Vittorio Giovannetti.
\newblock Quantum channels and their entropic characteristics.
\newblock {\em Reports on Progress in Physics}, 75(4):046001, April 2012.
\newblock arXiv:1202.6480.

\bibitem{HW01}
Alexander~S. Holevo and Reinhard~F. Werner.
\newblock Evaluating capacities of bosonic {Gaussian} channels.
\newblock {\em Physical Review A}, 63:032312, February 2001.
\newblock arXiv:quant-ph/9912067.

\bibitem{KS1}
Robert Koenig and Graeme Smith.
\newblock The entropy power inequality for quantum systems.
\newblock {\em IEEE Transactions on Information Theory (to be published)}, May
  2012.
\newblock arXiv:1205.3409.

\bibitem{KS3}
Robert Koenig and Graeme Smith.
\newblock Classical capacity of quantum thermal noise channels to within
  1.45~bits.
\newblock {\em Physical Review Letters}, 110:040501, January 2013.
\newblock arXiv:1207.0256.

\bibitem{KS2}
Robert Koenig and Graeme Smith.
\newblock Limits on classical communication from quantum entropy power
  inequalities.
\newblock {\em Nature Photonics}, 7:142--146, 2013.
\newblock arXiv:1205.3407.

\bibitem{KW09}
Robert Koenig and Stephanie Wehner.
\newblock A strong converse for classical channel coding using entangled
  inputs.
\newblock {\em Physical Review Letters}, 103:070504, August 2009.
\newblock arXiv:0903.2838.

\bibitem{KS4}
Robert Koenig, Stephanie Wehner, and J\"urg Wullschleger.
\newblock Unconditional security from noisy quantum storage.
\newblock {\em IEEE Transactions on Information Theory}, 58:1962--1984, 2012.
\newblock arXiv:0906.1030.

\bibitem{Lupo2011}
Cosmo Lupo, Stefano Pirandola, Paolo Aniello, and Stefano Mancini.
\newblock On the classical capacity of quantum gaussian channels.
\newblock {\em Physica Scripta}, T143:014016, 2011.
\newblock arXiv:1012.5965v2.

\bibitem{N99}
Ashwin Nayak.
\newblock Optimal lower bounds for quantum automata and random access codes.
\newblock In {\em Proceedings of the 40th Annual Symposium on Foundations of
  Computer Science}, pages 369--376, New York City, NY, USA, October 1999.
\newblock arXiv:quant-ph/9904093.

\bibitem{Ogawa}
Tomohiro Ogawa and Hiroshi Nagaoka.
\newblock Strong converse to the quantum channel coding theorem.
\newblock {\em IEEE Transactions on Information Theory}, 45:2486--2489,
  November 1999.
\newblock arXiv:quant-ph/9808063.

\bibitem{ON07}
Tomohiro Ogawa and Hiroshi Nagaoka.
\newblock Making good codes for classical-quantum channel coding via quantum
  hypothesis testing.
\newblock {\em IEEE Transactions on Information Theory}, 53(6):2261--2266, June
  2007.

\bibitem{RennerThesis}
Renato Renner.
\newblock {\em Security of Quantum Key Distribution}.
\newblock PhD thesis, ETH Z\"urich, December 2005.
\newblock arXiv:quant-ph/0512258.

\bibitem{RS}
Renato Renner and Robert Koenig.
\newblock {\em Universally Composable Privacy Amplification Against Quantum
  Adversaries}, volume 3378 of {\em Lecture Notes in Computer Science}.
\newblock Springer Berlin Heidelberg, 2005.
\newblock arXiv:quant-ph/0403133.

\bibitem{RWISIT1}
Renato Renner and Stefan Wolf.
\newblock Smooth {R\'enyi} entropy and applications.
\newblock In {\em Proceedings of the 2007 International Symposium on
  Information Theory}, page 232, 2004.

\bibitem{SW97}
Benjamin Schumacher and Michael~D. Westmoreland.
\newblock Sending classical information via noisy quantum channels.
\newblock {\em Physical Review A}, 56:131--138, July 1997.

\bibitem{T12}
Terence Tao.
\newblock {\em Topics in Random Matrix Theory}, volume 132 of {\em Graduate
  Studies in Mathematics}.
\newblock American Mathematical Society, 2012.
\newblock see also
  http://terrytao.wordpress.com/2010/01/03/254a-notes-1-concentration-of-measure.

\bibitem{Tomamichelthesis}
Marco Tomamichel.
\newblock {\em A Framework for Non-Asymptotic Quantum Information Theory}.
\newblock PhD thesis, ETH Z\"urich, March 2012.
\newblock arXiv:1203.2142.

\bibitem{WPGCRSL12}
Christian Weedbrook, Stefano Pirandola, Ra\'{u}l Garc\'{\i}a-Patr\'{o}n,
  Nicolas~J. Cerf, Timothy~C. Ralph, Jeffrey~H. Shapiro, and Seth Lloyd.
\newblock Gaussian quantum information.
\newblock {\em Reviews of Modern Physics}, 84:621--669, May 2012.
\newblock arXiv:1110.3234.

\bibitem{StrongConversePureLoss}
Mark~M. Wilde and Andreas Winter.
\newblock Strong converse for the classical capacity of the pure-loss bosonic
  channel.
\newblock {\em Problems of Information Transmission (to be published)}, August
  2013.
\newblock arXiv:1308.6732.

\bibitem{Entanglementbreaking}
Mark~M. Wilde, Andreas Winter, and Dong Yang.
\newblock Strong converse for the classical capacity of entanglement-breaking
  and {Hadamard} channels.
\newblock June 2013.
\newblock arXiv:1306.1586.

\bibitem{W99}
Andreas Winter.
\newblock Coding theorem and strong converse for quantum channels.
\newblock {\em IEEE Transactions on Information Theory}, 45(7):2481--2485,
  1999.

\bibitem{Wolfowitz1964}
Jacob Wolfowitz.
\newblock {\em Coding Theorems of Information Theory}, volume~31.
\newblock Springer, 1964.

\bibitem{MAC}
Brent~J. Yen and Jeffrey~H. Shapiro.
\newblock Multiple-access bosonic communications.
\newblock {\em Physical Review A}, 72:062312, December 2005.
\newblock arXiv:quant-ph/0506171.

\end{thebibliography}

\end{document}